\documentclass[11pt]{article}
	\usepackage[margin=1in]{geometry}

	\usepackage{amsmath, amssymb, amsthm, mathtools, amsfonts}
	\usepackage{bm}
	\usepackage{physics}
	\usepackage{array}
	\usepackage{mathrsfs}
	\usepackage{bbm}
	\usepackage{pifont}
	\usepackage{xcolor}
	
	\newtheorem{corollary}{Corollary}
	\newtheorem{proposition}{Proposition}
	\newtheorem{definition}{Definition}
	\newtheorem{lemma}{Lemma}
	\newtheorem{assumption}{Assumption}
	\theoremstyle{definition}
	\newtheorem{example}{Example}
	\newtheorem{remark}{Remark}
	
	\newcommand{\C}{\mathbb{C}}
	
	\newcommand{\RR}{\mathbb{R}}

	\DeclareMathOperator{\diag}{diag}
	\newcommand{\Spin}{\mathrm{Spin}}

	\newcommand{\Cl}{\mathrm{Cl}}

	\newcommand{\id}{\mathrm{id}}

	\providecommand{\CC}{\mathbb{C}}
	\providecommand{\DD}{\mathbb{D}}
	
	\providecommand{\id}{\mathrm{id}}
	\providecommand{\End}{\mathrm{End}}
	\DeclareMathOperator{\Sym}{Sym}

	\usepackage[authoryear,round]{natbib}

	\usepackage{hyperref}
	
	\title{Information Geometry of Bounded Rationality:\\
		\Large 
			 Entropy--Regularised Choice with Hyperbolic and Elliptic Quantum Geometries}
				
	\author{Anders Karlstr\" om, Christer Persson\\{\scriptsize KTH Royal Institute of Technology }}

\begin{document}
		
	\maketitle

	\begin{abstract}
		Models of bounded rationality span a variety of formalisms.  Two prominent softened--choice paradigms are
		quantum--like (QL) models, which use Hilbert--space amplitudes to account for context and order effects, and
		entropy--regularised (ER) models, including rational inattention, which modify expected utility by adding an
		information cost.  This paper develops a single information--geometric framework in which these families arise
		as different manifestations of the same underlying geometric structure on the probability simplex.
		
		Starting from the Fisher--Rao geometry of the open simplex $\Delta^{n-1}$, we introduce \emph{least--action rationality}
		(LAR) as a variational principle for decision dynamics, formulated in amplitude (square--root) coordinates and then
		lifted to the cotangent phase space $N:=T^*\RR^n$ of unnormalised amplitudes.  The lift carries its canonical symplectic
		form and, from the underlying Hessian data, a canonical para--K\"ahler (neutral) structure.  In the linear evaluator case
		$\widehat V=\widehat S+\widehat F$ with $\widehat S^\top=\widehat S$ and $\widehat F^\top=-\widehat F$, the lifted dynamics
		separate an evaluative (potential) channel generated by $\widehat S$ from a comparative (co--utility) channel generated by
		$\widehat F$, the latter extending Fishburn--type skew--symmetric bilinear regret.  On the distinguished zero--residual
		Lagrangian leaf, the flow admits a compact split--complex (Schr\"odinger--type) representation, and observable choice
		probabilities are obtained by the hyperbolic Born--type normalisation.
		
		Boundedly rational behaviour appears when the latent lifted dynamics are reduced back to epistemic motion on the Fisher--Rao
		simplex.  The induced simplex preference one--form admits a canonical exact/complementary decomposition with respect to the
		Fisher--Rao metric: an integrable utility component $dU$ and a divergence--free co--utility component whose curvature $d\beta$
		measures the resulting path dependence (holonomy) of preferences.  Further loss of latent structure through projection and
		context--dependent readouts renders the induced simplex process generically contextual, producing order effects, violations of
		the law of total probability, and interference--like terms as systematic shadows of the underlying rational flow.
		
		Finally, we show how standard complex (elliptic) quantum dynamics arises from the same real symplectic phase space by imposing
		additional modelling input: a K\"ahler polarisation (equivalently, a complex Lagrangian eigenbundle admissibility restriction)
		that replaces the full Hamilton equation by a projected Hamilton evolution.  In this sense, unitary quantum dynamics is not a
		primitive postulate but a coherent restriction of the underlying least--action framework.
	\end{abstract}

\clearpage           
\setcounter{page}{1}  

\section{Introduction}

Models of bounded rationality often describe behaviour through probabilistic rather than deterministic choice.
In discrete choice, such behaviour is represented by a probability vector
\(q=(q_1,\dots,q_n)\in\Delta^{n-1}\), interpreted either as a stochastic choice rule in the sense of
\cite{Luce1959} or as a von~Neumann--Morgenstern lottery over \(n\) outcomes \citep{vonNeumannMorgenstern1944}.
Two influential paradigms build on this representation in markedly different styles.
Quantum--like cognitive (QL) models represent behaviour through amplitudes and contextual projections, whereas
entropy--regularised (ER) models derive smooth probabilistic choice rules directly on the simplex.
The aim of this paper is to show that \emph{information geometry} provides a natural common framework for these approaches.
Moreover, we show that the same geometric lift underlying ER models also canonically induces a \emph{hyperbolic} (split--signature) quantum--like structure.

Quantum--like cognitive (QL) models describe the cognitive state as an amplitude---either a state vector or a
density operator---in a (typically complex) Hilbert space.  Deliberation and contextual change are represented
by unitary (or more generally completely positive) maps, and choice probabilities follow from the Born rule
(e.g.\ \citet{Aerts1995,BusemeyerEtAl2006,Aerts2009,BusemeyerBruza2012,YearsleyBusemeyer2016,Yearsley2017}).
Because contextual projections and observables need not commute, QL models naturally generate several empirical
signatures of bounded rationality, such as context effects, order effects, and deviations from the classical
law of total probability (e.g.\ \citet{Khrennikov2005,KhrennikovHaven2009,BruzaKittoNelsonMcEvoy2009,PothosBusemeyer2013,HavenKhrennikov2016}).
In this perspective bounded rationality reflects a contextual, noncommutative calculus, rather than noisy
perturbations of an underlying classical utility model.\footnote{The ``quantum'' label here refers to the use
	of Hilbert--space contextuality and interference, rather than to any claim of microscopic quantum computation in
	neural tissue \citep{deBarrosSuppes2009}.}

Entropy--regularised (ER) models provide smooth probabilistic choice rules directly on the simplex.
A canonical example is Rational Inattention, in which a decision maker trades off deterministic utility against
an information cost.  The Luce--logit (softmax) rule arises when this cost is Shannon entropy, while more
general variational formulations allow broader classes of information costs \citep{Sims2003,MatejkaMcKay2015}.
More broadly, information--theoretic ideas relating probabilistic choice, entropy, and minimum information have
appeared in the econometric and psychological literatures in several forms.\footnote{See, e.g.,
	\cite{McFadden1974,SnickarsWeibull1977,Luce2003}.}  In all cases, the resulting choice rules reflect a balance
between preference and information cost.

Although QL and ER models are usually presented as distinct paradigms, they admit a common interpretation in
information geometry.
On the ER side, entropy penalties and exponential--family structure arise from convex duality on the Fisher--Rao
statistical manifold \((\Delta^{n-1}_{>0},g_F)\).\footnote{See \cite{Ly2017} for an exposition of Fisher--Rao
	geometry in mathematical psychology.}
On the QL side, a closely related geometric structure appears after the amplitude reparametrisation
\(q\mapsto\rho\) given by \(\rho_k=\sqrt{q_k}\).  This map identifies the simplex interior with the positive
orthant of the unit sphere, pulls back \(g_F\) (up to scale) to the round metric, and places the amplitude
description in the standard symplectic/Hamiltonian framework familiar from geometric formulations of quantum
mechanics \citep{Kibble1979,Heslot1985,BrodyHughston2001,Goyal2010,ReginattoHall2012,Reginatto2013,Molitor2012,
	Caticha2021}.  Under this reparametrisation, expected utilities become quadratic forms in amplitudes
\citep{AertsHavenSozzo2018}, and the dynamics of probabilities may be analysed via an induced dynamics on
amplitudes.

Our starting point is this shared information--geometric foundation.  We formulate a local
\emph{least--action rationality} (LAR) principle for decision dynamics, in the spirit of the differential
approach to preferences \citep{Machina1982,Russell1991} and related least--action proposals in cognitive
modelling (cf.\ \citet{BettiGori2016LeastCognitiveAction,FoxKotelba2018PLPA}), but applied at the level of
\emph{amplitudes} rather than directly on the simplex.
At each state, the realised velocity is the metric projection (in the Fisher--Rao geometry induced by the
amplitude embedding) of a preferred direction specified by a preference covector field.
In the linear case this yields an evolution for \emph{unnormalised} amplitudes
\(\dot{\tilde\rho}=\widehat V\tilde\rho\), where the generator decomposes canonically as
\(\widehat V=\widehat S+\widehat F\) with \(\widehat S^{\!\top}=\widehat S\) (evaluative channel) and
\(\widehat F^{\!\top}=-\widehat F\) (co--utility/rotational channel).

Passing from amplitudes to their cotangent lift yields an ambient first--order Hamiltonian system on
\(\mathbb R^{2n}\cong T^*\mathbb R^n\).
The lifted Fisher--Rao structure canonically induces a flat para--K\"ahler geometry on this space (neutral metric,
symplectic form, and product structure).
On a distinguished \emph{zero--residual} sector (where the least--action projection constraint is exactly met),
the dynamics admit a split--complex (\emph{hyperbolic}) Schr\"odinger--type representation.

This para--Hermitian lift is rich enough to reproduce several signatures typically attributed to quantum--like
models. It yields a hyperbolic Born--type readout, permits noncommutative observables, and generates phase
sensitivity, interference patterns, and contextual effects, all without positing a complex Hilbert space.
Moreover, the hyperbolic geometry provides a canonical neutral quadratic index whose evolution obeys an exact
balance law with a nonnegative source term, thereby organising episodes by the cumulative accumulation of
least--action residual.

Although QL models are often formulated in complex Hilbert space, it is frequently acknowledged---and
sometimes adopted---that real--valued formulations already capture the essential contextual--projection
structure without appealing to complex phases (e.g.\ \cite{BruzaKittoNelsonMcEvoy2009}).
The para--K\"ahler framework developed here provides a principled foundation for such real formulations by
deriving a genuinely hyperbolic (split--signature) quantum geometry directly from classical Fisher--Rao
decision dynamics, connecting also to earlier work on hyperbolic quantum mechanics
\citep{Khrennikov2000hyperbolicquantummechanics}.
In this view, the split--complex notation used later is a convenient packaging of an underlying real geometry,
not an approximation to an elliptic (complex) quantum model.

QL models with coherent (unitary) evolution form a different and equally important strand of the QL
literature (e.g.\ \cite{BusemeyerWangLambertMogiliansky2009,BusemeyerZhangBalakrishnanWang2020}).
We show in \autoref{sec:4} how such \emph{elliptic} (Riemannian/complex) quantum dynamics arises
in our framework only after an additional coherence input: one must impose an admissibility restriction that
selects a holomorphic Lagrangian sector (a polarisation) compatible with the symplectic structure, so that the
projected dynamics become unitary on the corresponding complex state space.

The rest of the paper is organised as follows.
\autoref{sec:2} formulates least--action rationality on the probability simplex from three geometric
ingredients: Fisher--Rao distinguishability, a preference covector field, and a least--action variational
principle.
\autoref{sec:3} lifts this structure to an ambient amplitude model on \(\mathbb R^{2n}\): we derive the
quadratic LAR Lagrangian, the associated Hamiltonian system, its canonical para--K\"ahler geometry, and the
split--complex (hyperbolic) Schr\"odinger representation on the distinguished zero--residual sector.
\autoref{sec:4} then shows how standard (elliptic) quantum geometry arises by imposing a
polarised admissibility restriction, making precise in what sense coherence--preserving unitary dynamics is a
special case of the lifted theory.
\autoref{sec:BR} returns to behavioural implications, analysing how the para--K\"ahler lift
organises systematic deviations from expected utility.
\autoref{sec:conclusions} concludes and discusses directions for future work.

\section{A least--action rationality principle}
\label{sec:2}
In this section we begin with making explicit the three foundational assumptions that constitute a principle of rationality used in this paper. Each assumption
isolates a specific geometric ingredient: (i) a notion of
\emph{distinguishability} between nearby redistributions of a lottery,
(ii) an \emph{evaluative} (preference) structure defined on such
redistributions, and (iii) a \emph{variational} principle.

\subsection{Distinguishability}
\label{subsec:manifold}

We take as primitive observable objects lotteries over the finite outcome set $\{1,\dots,n\}$, represented by probability
vectors. Accordingly, we consider the open probability simplex
\[
\mathcal M:=\Delta^{n-1}_{>0}
=\Bigl\{\, q=(q_1,\dots,q_n)\in\RR^n \ \Big|\ q_i>0,\ \sum_{i=1}^n q_i=1 \,\Bigr\},
\]
which is a smooth embedded submanifold of the affine hyperplane $\{q\in\RR^n:\sum_i q_i=1\}$ of dimension $n-1$. Its closure
is the closed simplex
\[
\Delta^{n-1}
=\Bigl\{\, q\in\RR^n \ \Big|\ q_i\ge 0,\ \sum_{i=1}^n q_i=1 \,\Bigr\},
\]
so that $\mathcal M$ is the interior of $\Delta^{n-1}$. For each $q\in\mathcal M$ the tangent space is canonically
identified with the codimension--one subspace
\[
T_q\mathcal M=\Bigl\{\, v\in\RR^n \ \Big|\ \sum_{i=1}^n v_i=0 \,\Bigr\}.
\]

To formalise ``distinguishability'' we endow lotteries with a perceptual metric.  We obtain it by lifting lotteries to the
unit sphere, where the round geometry provides a canonical notion of distance, and then pulling this distance back to the simplex. Let
$\langle\cdot,\cdot\rangle$ be the Euclidean inner product on $\RR^n$ and $\|x\|:=\sqrt{\langle x,x\rangle}$ its norm.
Given any nonzero \emph{unnormalised amplitude} $\tilde\rho\in\RR^n\setminus\{0\}$, set
\[
\rho:=\frac{\tilde\rho}{\|\tilde\rho\|}\in\mathbb S^{n-1},
\qquad
\mathbb S^{n-1}:=\{\,\rho\in\RR^n:\|\rho\|=1\,\}.
\]
We read out an observable lottery by componentwise squaring,
\[
q_i=\rho_i^2=\frac{\tilde\rho_i^{\,2}}{\|\tilde\rho\|^2}\qquad(i=1,\dots,n),
\]
equivalently via the smooth surjection
\[
\pi:\mathbb S^{n-1}\to\Delta^{n-1},
\qquad
\pi(\rho)=(\rho_1^2,\dots,\rho_n^2)=\rho^{\odot 2},
\]
where $x^{\odot 2}$ denotes componentwise squaring. The map $\pi$ is sign--blind: for $q\in\mathcal M$ the fibre
$\pi^{-1}(q)$ consists of the $2^n$ sign choices $\rho\mapsto s\rho$ with $s=\mathrm{diag}(\pm1)$, whereas on faces the same
action produces degeneracies when some components vanish. For geometric constructions on lotteries it is convenient to choose
the canonical square--root section into the nonnegative orthant
\[
\overline{\mathbb S}^{\,n-1}_+:=\{\,\rho\in\mathbb S^{n-1}:\rho_i\ge 0\ \forall i\,\},
\qquad
\iota:\Delta^{n-1}\to\overline{\mathbb S}^{\,n-1}_+,
\qquad
\iota(q)=\sqrt q:=(\sqrt{q_1},\dots,\sqrt{q_n}),
\]
so that $\pi\circ\iota=\mathrm{id}_{\Delta^{n-1}}$. On the interior $\mathcal M$ the restriction of $\iota$ is smooth and
maps into the open orthant
\[
\mathbb S^{n-1}_+:=\{\,\rho\in\mathbb S^{n-1}:\rho_i>0\ \forall i\,\}.
\]
We emphasise that $\iota$ serves only as a smooth lift (equivalently, a diffeomorphism
$\iota:\mathcal M\to\mathbb S^{n-1}_+$ with inverse $\pi\!\mid_{\mathbb S^{n-1}_+}$) for expressing simplex--level geometric objects;
it does not constrain the sign pattern of the underlying normalised amplitude $\tilde\rho/\|\tilde\rho\|$, which remains unconstrained.

Let $g_{\mathbb S}$ denote the round metric on $\mathbb S^{n-1}$ induced by $\langle\cdot,\cdot\rangle$, and let $d_{\mathbb S}$
be the associated geodesic distance. Diagonal sign flips $s=\mathrm{diag}(\pm 1)\in O(n)$ act by isometries on
$(\mathbb S^{n-1},g_{\mathbb S})$, hence any sign section $q\mapsto s\sqrt q$ induces the same pullback metric on $\mathcal M$.

We now postulate that perceptual distance between lotteries is the angular distance between their square--root lifts on the sphere, up to an overall scale.

\begin{assumption}[Distinguishability (A1)]\label{assump:distinguishability}
	Fix a constant $c>0$ and define the perceptual distance between lotteries $q,q'\in\Delta^{n-1}$ by
	\[
	d_{\mathrm{perc}}(q,q')
	:=c\,d_{\mathbb S}\!\bigl(\iota(q),\iota(q')\bigr)
	=c\,\arccos\!\Bigl(\bigl\langle \iota(q),\iota(q')\bigr\rangle\Bigr)
	=c\,\arccos\!\Bigl(\sum_{i=1}^n \sqrt{q_iq'_i}\Bigr).
	\]
\end{assumption}
Since $\iota$ is injective, $d_{\mathrm{perc}}$ is the restriction of the spherical geodesic distance to the embedded set
$\iota(\Delta^{n-1})\subset\mathbb S^{n-1}$, hence a metric on $\Delta^{n-1}$. Moreover, because
$\iota(\Delta^{n-1})\subset\overline{\mathbb S}^{\,n-1}_+$, one has $0\le d_{\mathrm{perc}}(q,q')\le c\pi/2$.

The definition above provides a global perceptual metric on $\Delta^{n-1}$. On the interior $\mathcal M$ we refine it
infinitesimally to obtain a local distinguishability tensor by differentiating the lift $\iota$. Since the square--root
map becomes singular as some coordinates approach zero, this refinement is asserted on $\mathcal M$ and must be interpreted
piecewise if a trajectory contacts the boundary.

\begin{remark}[Boundary contacts: piecewise interior reading]\label{rem:piecewise-interior-dist}
	The map $q\mapsto\sqrt q$ is smooth on $\mathcal M$ but becomes singular as some $q_i\to 0$. Accordingly, all differential
	identities derived from $d\iota$ are understood on $\mathcal M$ and, if a trajectory reaches the boundary, are to be read
	piecewise on each open time interval during which $q(t)$ remains in the relative interior of a fixed face (restrict to the
	active coordinates). Switching times are treated by one--sided limits; no additional boundary rule is imposed.
\end{remark}

For $q\in\mathcal M$ and $v,w\in T_q\mathcal M$, writing $\rho=\iota(q)=\sqrt q$, differentiation of $q_i=\rho_i^2$ gives
$dq_i=2\rho_i\,d\rho_i$, hence
\[
d\rho_i(v)=\frac{v_i}{2\rho_i},
\qquad
\bigl(\iota^*g_{\mathbb S}\bigr)(q)[v,w]
=\frac14\sum_{i=1}^n\frac{v_iw_i}{q_i}.
\]
It follows that
\[
\iota^*g_{\mathbb S}=\tfrac14\,g_F,
\qquad
g_F(q)[v,w]:=\sum_{i=1}^n\frac{v_iw_i}{q_i}\quad\text{for all }v,w\in T_q\mathcal M,
\]
where $g_F$ is the Fisher--Rao metric on $\mathcal M$. 
Thus the induced local distinguishability metric on lotteries is (up to the constant factor $1/4$) the Fisher--Rao metric; we take $g_F$ as the canonical simplex--level tensor associated with A1.

The constant $c$ in Assumption~\ref{assump:distinguishability}
sets the overall units of distance; infinitesimally, the induced quadratic form is proportional to $\iota^*g_{\mathbb S}$,
namely $c^2\,\iota^*g_{\mathbb S}=\tfrac{c^2}{4}\,g_F$ on $T_q\mathcal M$. In particular, $g_F$ weights perturbations in
rare outcomes more heavily via the factors $1/q_i$, and becomes singular at the boundary of $\Delta^{n-1}$, motivating our
use of the interior $\mathcal M$ for smooth Riemannian geometry.

Finally, for any smooth trajectory $t\mapsto q(t)\in\mathcal M$, the lifted curve $\rho(t):=\iota(q(t))\in\mathbb S^{n-1}_+$
is an auxiliary geometric lift and satisfies
\[
\|\dot q(t)\|_{g_F}^2=4\,\|\dot\rho(t)\|_{g_{\mathbb S}}^2.
\]

\begin{remark}
	In psychophysics the same local \textit{discriminability} index is modelled by the quadratic form
	$v\mapsto \sum_{i=1}^n v_i^2/q_i$ on $T_q\mathcal M$ (equivalently $v\mapsto \|d\iota_q(v)\|_{g_{\mathbb S}}^2$ under
	$\iota(q)=\sqrt q$); see \citep{Stevens1957,Sakitt1973,Siomopoulos1975}. In this paper we use the information-theoretic
	terminology, where the Fisher--Rao metric $g_F$ measures local statistical \emph{distinguishability}.
\end{remark}

\subsection{Preference covector field}
\label{subsec:PreferenceOperators-full}

As emphasised already by~\cite{Russell1991}, differential geometry provides a natural language for preferences:
marginal evaluation is encoded by a covector field (a $1$--form).
In our framework, this covector field is induced by a fixed \emph{preference operator}
$\widehat V\in\End(\RR^n)$ acting on latent amplitudes.

Let $g_{\mathbb S}$ be the round metric on the amplitude sphere
$\mathbb S^{n-1}\subset\RR^n$ induced by the Euclidean inner product $\langle\cdot,\cdot\rangle$.
For $\rho\in\mathbb S^{n-1}$ denote by
\[
Q_\rho:=I-\rho\rho^\top
\]
the orthogonal projector onto the tangent space
\[
T_\rho\mathbb S^{n-1}=\{\eta\in\RR^n:\langle\rho,\eta\rangle=0\}.
\]
A (spherical) preference covector is a $1$--form $\bar\alpha^{\mathbb S}\in\Omega^1(\mathbb S^{n-1})$ assigning to each
infinitesimal adjustment $\eta\in T_\rho\mathbb S^{n-1}$ an instantaneous evaluative score
$\bar\alpha^{\mathbb S}_\rho(\eta)\in\RR$.
We restrict attention to the ambient--operator class: fix $\widehat V\in\End(\RR^n)$ and define
\begin{equation}\label{eq:pref-1form}
	\bar\alpha^{\mathbb S}_\rho(\eta)
	=\langle \eta,\,Q_\rho\,\widehat V\rho\rangle,
	\qquad \eta\in T_\rho\mathbb S^{n-1}.
\end{equation}
Since $\eta\perp\rho$, we have $\langle \eta,Q_\rho \widehat V\rho\rangle=\langle \eta,\widehat V\rho\rangle$.
Thus $\bar\alpha^{\mathbb S}$ is the $g_{\mathbb S}$--musical dual of the tangent vector field
\begin{equation}\label{eq:pref-vectorfield}
	X_{\widehat V}(\rho):=Q_\rho\,\widehat V\rho\in T_\rho\mathbb S^{n-1},
	\qquad
	\bar\alpha^{\mathbb S}=g_{\mathbb S}(X_{\widehat V},\cdot).
\end{equation}
Within this class, $\bar\alpha^{\mathbb S}$ is invariant under the shift $\widehat V\mapsto \widehat V+\lambda I$,
since $Q_\rho(\lambda\rho)=0$. Hence only the class $[\widehat V]\in\End(\RR^n)/(\RR\cdot I)$ is observable at the
spherical level; when convenient we fix a representative by imposing $\tr(\widehat V)=0$.

Observable lotteries live on the interior simplex
\[
\mathcal M=\Delta^{n-1}_{>0}
=\Bigl\{q\in\RR^n_{>0}:\sum_{i=1}^n q_i=1\Bigr\}.
\]
We represent $q\in\mathcal M$ by its positive square root
\[
\iota:\mathcal M\to\mathbb S^{n-1}_{+},
\qquad
\iota(q)=\rho=\sqrt q,
\]
so that $\rho_i=\sqrt{q_i}$ and $\sum_i\rho_i^2=1$.
Restricting $\bar\alpha^{\mathbb S}$ to the open set $\mathbb S^{n-1}_{+}$ and pulling back along $\iota$ yields the
\emph{marginal preference covector} on lotteries:
\[
\mathcal T:\End(\RR^n)\to\Omega^1(\mathcal M),
\qquad
\mathcal T(\widehat V):=\beta:=\iota^*\!\bigl(\bar\alpha^{\mathbb S}\!\mid_{\mathbb S^{n-1}_{+}}\bigr).
\]
Here $\iota$ is used only as a smooth lift (equivalently, a diffeomorphic identification of $\mathcal M$ with $\mathbb S^{n-1}_+$); it does not constrain the sign pattern of any latent amplitude trajectory.

Decompose
\[
\widehat V=\widehat S+\widehat F,
\qquad
\widehat S^\top=\widehat S,
\qquad
\widehat F^\top=-\widehat F.
\]
The central structural fact proved below is that $\widehat S$ induces an exact (utility) component on $\mathcal M$,
while $\widehat F$ induces a purely circulatory component which is co--exact in Fisher--Rao geometry.

\begin{assumption}[Preference structure]\label{assump:preference}
	The DM's instantaneous evaluation of infinitesimal changes on the amplitude sphere is encoded by a spherical $1$--form of
	ambient--operator type:
	\begin{enumerate}
		\item[\textnormal{(A2)}] There exists a fixed $\widehat V\in\End(\RR^n)$ such that for every
		$\rho\in\mathbb S^{n-1}$ and $\eta\in T_\rho\mathbb S^{n-1}$,
		\[
		\bar\alpha^{\mathbb S}_\rho(\eta)=\langle \eta,\,Q_\rho\widehat V\rho\rangle,
		\qquad Q_\rho:=I-\rho\rho^\top .
		\]
		\item[\textnormal{(A2$'$)}] Decompose $\widehat V=\widehat S+\widehat F$ with
		$\widehat S^\top=\widehat S$ and $\widehat F^\top=-\widehat F$. Passing to $\rho\in\mathbb S^{n-1}$,
		only the induced tangential component is behaviourally relevant (i.e.\ relevant for $\bar\alpha^{\mathbb S}$), i.e.
		\[
		\bar\alpha^{\mathbb S}_\rho(\eta)=\langle \widehat V\rho,\eta\rangle
		\qquad(\eta\in T_\rho\mathbb S^{n-1}),
		\]
		equivalently
		\[
		\bar\alpha^{\mathbb S}
		=\bigl((Q_\rho\widehat V\rho)^\top d\rho\bigr)\big|_{T\mathbb S^{n-1}}
		\quad(\text{or }(\widehat V\rho)^\top d\rho\big|_{T\mathbb S^{n-1}}).
		\]
	\end{enumerate}
	The form $\bar\alpha^{\mathbb S}$ is invariant under the additive shift $\widehat V\mapsto \widehat V+\lambda I$
	(since $Q_\rho\rho=0$), hence only the class $[\widehat V]\in\End(\RR^n)/(\RR\cdot I)$ is observable.
	When convenient we adopt the trace--free convention $\tr(\widehat V)=0$ (equivalently $\tr(\widehat S)=0$).
\end{assumption}

The equivalence of \textnormal{(A2)} and \textnormal{(A2$'$)} is immediate from tangency: if $\eta\in T_\rho\mathbb S^{n-1}$
then $Q_\rho\eta=\eta$, hence
\[
\bar\alpha^{\mathbb S}_\rho(\eta)
=\langle \eta,\,Q_\rho\widehat V\rho\rangle
=\langle Q_\rho\eta,\,\widehat V\rho\rangle
=\langle \eta,\,\widehat V\rho\rangle,
\]
so $\bar\alpha^{\mathbb S}=(\widehat V\rho)^\top d\rho=(Q_\rho\widehat V\rho)^\top d\rho$.

\begin{remark}[Boundary contacts: stratified simplex reading]\label{rem:stratified-reading}
The observable lottery is read out from an unnormalised amplitude $\tilde\rho\in\RR^n\setminus\{0\}$ by
$q_i=\tilde\rho_i^2/\|\tilde\rho\|^2$ (equivalently, $q=\pi(\rho)=\rho^{\odot 2}$ for $\rho=\tilde\rho/\|\tilde\rho\|\in\mathbb S^{n-1}$).
The amplitude--level dynamics introduced later is formulated for unnormalised amplitudes $\tilde\rho\in\RR^n$.
Thus crossings of coordinate hyperplanes $\tilde\rho_i=0$ are not singular for the amplitude dynamics; they only imply that the
readout $q_i=\tilde\rho_i^2/\|\tilde\rho\|^2$ touches the boundary.
However, the square--root chart
$q\mapsto\sqrt q$ and the Fisher--Rao tensor are smooth only on the relative interior of a face of $\Delta^{n-1}$.
Accordingly, any simplex--level differential identity derived via $d\iota$ or $g_F$ is to be read \emph{piecewise}:
for a trajectory $q(t)$, fix any open time interval $I$ on which the support
$J=\{i:q_i(t)>0\}$ is constant.  On $I$ we identify $q_J(t)$ with a point of the lower--dimensional interior simplex
$\Delta^{|J|-1}_{>0}$ (active coordinates) and apply all formulas there.  Switching times (when some $q_i(t)=0$) are treated
by one--sided limits; no additional boundary rule is imposed.
For statements relying on $H^1(\mathbb S^{m-1})=0$ (e.g.\ the co--exact part below), the active set must satisfy $|J|\ge3$;
when $|J|=2$ a harmonic remainder may appear.
\end{remark}

\begin{proposition}[Utility and co--utility fields on the simplex]\label{prop:beta-decomp}
Let $\mathcal M=\Delta^{n-1}_{>0}$ and assume $n\ge3$.  The statements below are understood on $\mathcal M$ and,
when a trajectory touches the boundary, piecewise on each interval of constant support as in \autoref{rem:stratified-reading}.
	Let $\beta=\mathcal T(\widehat V)\in\Omega^1(\mathcal M)$ be the marginal preference covector induced by
	$\widehat V\in\End(\RR^n)$, and decompose $\widehat V=\widehat S+\widehat F$ with
	$\widehat S^\top=\widehat S$ and $\widehat F^\top=-\widehat F$. Define
	\[
	\beta^{\widehat S}:=\mathcal T(\widehat S),\qquad
	\beta^{\widehat F}:=\mathcal T(\widehat F),
	\qquad\text{so that}\qquad
	\beta=\beta^{\widehat S}+\beta^{\widehat F}.
	\]
	Then:
	\begin{enumerate}
		\item \emph{Utility (exact part).} There exists a smooth potential
		\begin{equation} \label{eq:U-def}
		U(q):=\tfrac12(\sqrt q)^\top \widehat S\,\sqrt q
		\end{equation}
		such that $\beta^{\widehat S}=dU$.
		
		\item \emph{Co--utility (co--exact part).} With respect to the Fisher--Rao metric $g_F$ on $\mathcal M$ (and on each
		fixed--support face with $|J|\ge3$ under \autoref{rem:stratified-reading}), there exists $\gamma\in\Omega^2(\mathcal M)$ such that
		$\beta^{\widehat F}=\delta_{g_F}\gamma$.
		Equivalently,
		\[
		\beta=dU+\delta_{g_F}\gamma .
		\]
	\end{enumerate}
	Define $\mathcal U:=dU$ and $\mathcal R:=\beta-\mathcal U$.
\end{proposition}

\begin{proof}
	On $\mathcal M$ the square--root map $\iota(q)=\rho=\sqrt q$ is smooth and identifies $\mathcal M$
	with the open set $\iota(\mathcal M)=\mathbb S^{n-1}_{+}\subset\mathbb S^{n-1}$. On $\mathbb S^{n-1}$ write
	\[
	\bar\alpha^{\mathbb S}=(\widehat V\rho)^\top d\rho
	=(\widehat S\rho)^\top d\rho+(\widehat F\rho)^\top d\rho
	=:\bar\alpha^{\widehat S}+\bar\alpha^{\widehat F}.
	\]
	Pull back along $\iota$ to $\mathcal M$ to obtain $\beta^{\widehat S}=\iota^*\bar\alpha^{\widehat S}$ and
	$\beta^{\widehat F}=\iota^*\bar\alpha^{\widehat F}$.
	
	\smallskip\noindent\emph{Utility part.}
	Let $\Phi(\rho):=\tfrac12\rho^\top\widehat S\rho$ on $\mathbb S^{n-1}$. Since $\widehat S$ is symmetric,
	$d\Phi=(\widehat S\rho)^\top d\rho=\bar\alpha^{\widehat S}$. Hence
	$\beta^{\widehat S}=d(\Phi\circ\iota)=dU$ with $U(q)=\tfrac12(\sqrt q)^\top\widehat S\,\sqrt q$.
	
	\smallskip\noindent\emph{Co--utility part.}
	If $\widehat F^\top=-\widehat F$, then $X(\rho):=\widehat F\rho$ is Killing on the round sphere, hence divergence-free,
	so its dual 1-form $\bar\alpha^{\widehat F}=g_{\mathbb S}(X,\cdot)$ is co-closed:
	$\delta_{g_{\mathbb S}}\bar\alpha^{\widehat F}=0$. For $n\ge3$, $H^1(\mathbb S^{n-1})=0$, so by Hodge theory on
	$\mathbb S^{n-1}$ there exists $\psi\in\Omega^2(\mathbb S^{n-1})$ with
	$\bar\alpha^{\widehat F}=\delta_{g_{\mathbb S}}\psi$.
	Restrict this identity to the open subset $\iota(\mathcal M)=\mathbb S^{n-1}_{+}\subset\mathbb S^{n-1}$ and pull back along
$\iota$.  Since $\iota:(\mathcal M,\iota^*g_{\mathbb S})\to(\mathbb S^{n-1}_{+},g_{\mathbb S})$ is an isometry, the
codifferential commutes with pullback on restricted forms, i.e.
$\iota^*(\delta_{g_{\mathbb S}}\psi)=\delta_{\iota^*g_{\mathbb S}}(\iota^*\psi)$.
Using $g_F=4\,\iota^*g_{\mathbb S}$ and the scaling rule $\delta_{\lambda g}=\lambda^{-1}\delta_g$ for $\lambda>0$, we obtain
\[
\beta^{\widehat F}=\iota^*\bar\alpha^{\widehat F}
=\delta_{\iota^*g_{\mathbb S}}(\iota^*\psi)
=\delta_{g_F}\bigl(4\,\iota^*(\psi|_{\mathbb S^{n-1}_{+}})\bigr).
\]
Thus $\beta^{\widehat F}=\delta_{g_F}\gamma$ with $\gamma:=4\,\iota^*(\psi|_{\mathbb S^{n-1}_{+}})$.

\end{proof}

\begin{remark}[Comparison with other operator models]
At this static, epistemic level the potential $U$ in~\eqref{eq:U-def}
encodes an operator-based expected-utility functional, where evaluation takes the
form of a quadratic expectation, as in \textit{projective expected utility (PEU)}  of ~\cite{LaMura2009}
and in the quantum expected-utility framework of~\cite{AertsHavenSozzo2018}.
 Our approach departs from these models at the latent level, where we keep track of the full decomposition
	$\widehat V=\widehat S+\widehat F$ and, in the (split-)complex lifts developed
	later, identify the antisymmetric part $\widehat F$ with the split-complex or
	complex sector of the underlying (para-)/K\"ahler geometry in
	Sections~\S\ref{sec:3} and~\S\ref{sec:4}. In particular, PEU is further discussed in Example \ref{example:PEU-EUT} below, see \autoref{subsec:R2n}.
\end{remark}

Before concluding this discussion of preference structure, it is useful to separate the symmetric and skew--symmetric parts
of the latent preference operator, $\widehat V=\widehat S+\widehat F$, and to discuss the distinct behavioural roles they play
when the induced marginal covector is expressed on $\mathcal M$ via the square--root chart.

Starting with the symmetric component $\widehat S^\top=\widehat S$, the induced simplex preference field is exact within the
linear lift class: it admits the potential
\[
U(q)=\tfrac12\,(\sqrt q)^\top \widehat S\,\sqrt q
=\tfrac12\sum_{i,j}\widehat S_{ij}\sqrt{q_i q_j}.
\]
This expression is quadratic in amplitudes (square--root probabilities).  When $\widehat S$ is diagonal, it reduces to the
standard additive expected--utility form $U(q)=\sum_i u_i q_i$ with $u_i=\tfrac12\widehat S_{ii}$, recovering classical
utility.  More generally, off--diagonal entries encode pairwise evaluative couplings between alternatives, weighted by belief
amplitudes rather than probabilities.  In this sense the symmetric potential provides a canonical pairwise--interaction
extension of additive utility; within the square--root chart, these are exactly the operator--induced utility potentials
(unique up to an additive constant, reflecting the $\RR\cdot I$ gauge): by Proposition~\ref{prop:beta-decomp}(1),
$\mathcal T(\widehat S)=dU$.

By contrast, the skew--symmetric component $\widehat F^\top=-\widehat F$ does not contribute to the exact (utility) channel
and instead generates circulation.  On the amplitude sphere it induces the divergence--free Killing field
$X_{\widehat F}(\rho)=\widehat F\rho$, whose dual $1$--form
$\bar\alpha^{\widehat F}(\rho)=(\widehat F\rho)^\top d\rho$ is co--closed.  Moreover, for any pair of amplitude states
$\rho,\rho'\in\mathbb S^{n-1}$ it defines the skew bilinear form
\[
R(\rho,\rho'):=\rho^\top \widehat F\,\rho',\qquad R(\rho,\rho')=-R(\rho',\rho).
\]
When restricted to the principal section $\rho=\sqrt q$ this yields the corresponding skew bilinear comparison
$R(\sqrt q,\sqrt{q'})$ on lotteries, matching the structure of Fishburn's skew--symmetric bilinear (SSB) preference form
in these coordinates \cite{Fishburn1984}.  Accordingly, $\widehat F$ may be interpreted as a continuous regret--type field,
encoding non--potential components of preference intensity that generate circulation rather than gradient flow.  This aligns
with Russell's differential--geometric view of regret as the nonconservative part of a preference field \citep{Russell1991},
where non--closed components can generate holonomy around loops.  In summary, $\widehat S$ governs the conservative evaluative
channel through a utility potential, whereas $\widehat F$ governs an intrinsic pairwise comparison (co--utility) channel that
can generate path dependence.

\subsection{Least--action rationality}
\label{subsec:least-action-rationality}

Given Assumption~A2 specifying a spherical preference covector field $\bar\alpha^{\mathbb S}\in\Omega^1(\mathbb S^{n-1})$, we
consider infinitesimal adjustments of an agent's state on $\mathbb S^{n-1}$.  When relating spherical objects to lotteries
$q\in\mathcal M$, we evaluate them on the principal section $\iota(q)=\sqrt q\in\mathbb S^{n-1}_+$, i.e.\ we use the restriction
$\alpha^{\mathbb S}:=\bar\alpha^{\mathbb S}\!\mid_{\mathbb S^{n-1}_+}$; this is a representative choice and does not constrain
the latent amplitude trajectory.

At $\rho\in\mathbb S^{n-1}$ admissible instantaneous displacements are tangent vectors $\dot\rho\in T_\rho\mathbb S^{n-1}$.
To compare such displacements with a (generally non-tangent) preferred direction in the ambient space, fix any smooth
\emph{ambient lift} $\widetilde X:\mathbb S^{n-1}\to\RR^n$ of the covector field, meaning
\[
\bar\alpha^{\mathbb S}_\rho(\delta\rho)=\langle \delta\rho,\,\widetilde X(\rho)\rangle
\qquad\text{for all }\rho\in\mathbb S^{n-1},\ \delta\rho\in T_\rho\mathbb S^{n-1}.
\]
(Such $\widetilde X$ always exists, and is unique up to addition of a radial term $\lambda(\rho)\rho$.)  Let
$Q_\rho:=I-\rho\rho^\top$ denote the orthogonal projector onto $T_\rho\mathbb S^{n-1}$.  Least--action (least--effort)
rationality postulates that the realised displacement is the closest admissible direction to $\widetilde X(\rho)$:
\[
\dot\rho
=\arg\min_{\xi\in T_\rho\mathbb S^{n-1}}\|\xi-\widetilde X(\rho)\|^2
=Q_\rho\,\widetilde X(\rho).
\]
Equivalently, since $Q_\rho\widetilde X(\rho)$ is the unique tangent vector satisfying
$\langle\delta\rho,Q_\rho\widetilde X(\rho)\rangle=\bar\alpha^{\mathbb S}_\rho(\delta\rho)$ for all
$\delta\rho\in T_\rho\mathbb S^{n-1}$, one has
\[
\dot\rho = g_{\mathbb S}^{-1}\bar\alpha^{\mathbb S}_\rho \in T_\rho\mathbb S^{n-1}.
\]

When $\bar\alpha^{\mathbb S}$ is induced by a linear evaluator in the ambient space,
\[
\bar\alpha^{\mathbb S}_\rho(\delta\rho)=\langle \delta\rho,\,Q_\rho\,\widehat V\rho\rangle
=\langle \delta\rho,\,\widehat V\rho\rangle,
\qquad \delta\rho\in T_\rho\mathbb S^{n-1},
\]
we may take $\widetilde X(\rho)=\widehat V\rho$, and the least--action direction becomes
\[
\dot\rho = Q_\rho\,\widehat V\rho.
\]
This defines the smooth vector field
\begin{equation}\label{eq:LAR-vector-field}
	X_{\widehat V}(\rho):=Q_\rho\,\widehat V\rho\in T_\rho\mathbb S^{n-1},
\end{equation}
which we refer to as the \emph{Least--Action Rationality (LAR) vector field}.

\begin{assumption}[Least--action rationality (LAR)]\label{assump:least-action}
	At each state $\rho\in\mathbb S^{n-1}$, fix an ambient representative $\widetilde X(\rho)\in\RR^n$ of the preference covector
	$\bar\alpha^{\mathbb S}_\rho$ as above.  The infinitesimal adjustment $\dot\rho\in T_\rho\mathbb S^{n-1}$ is the pointwise
	least--squares fit of $\widetilde X(\rho)$ by an admissible tangent vector:
	\[
	\dot\rho
	=\arg\min_{\xi\in T_\rho\mathbb S^{n-1}}\|\xi-\widetilde X(\rho)\|^2
	=Q_\rho\,\widetilde X(\rho).
	\]
	Equivalently, $\dot\rho=g_{\mathbb S}^{-1}\bar\alpha^{\mathbb S}_\rho$.
\end{assumption}

Restricting to the principal section $\iota(\mathcal M)\subset\mathbb S^{n-1}_+$, write
$\alpha^{\mathbb S}:=\bar\alpha^{\mathbb S}\!\mid_{\mathbb S^{n-1}_+}$ and recall that the induced simplex preference form is
$\beta=\iota^*\alpha^{\mathbb S}\in\Omega^1(\mathcal M)$.  On $\mathcal M$ define the Fisher--Rao natural--gradient field
$X_{\mathcal M}$ by
\[
g_F\bigl(X_{\mathcal M},\cdot\bigr)=\beta,
\qquad\text{where}\qquad
g_F=4\,\iota^*g_{\mathbb S}.
\]
Then for each $q\in\mathcal M$ one has the pushforward relation
\[
d\iota_q\bigl(X_{\mathcal M}(q)\bigr)
=\tfrac14\,g_{\mathbb S}^{-1}\alpha^{\mathbb S}_{\iota(q)}\ \in\ T_{\iota(q)}\mathbb S^{n-1},
\]
i.e.\ the simplex Fisher--Rao gradient flow is the $\iota^{-1}$--image of the spherical least--action direction field,
up to the constant time rescaling by a factor $4$.

\section{Cotangent lift and para-Kähler geometry}
\label{sec:3}

The previous section formulated Least--Action Rationality (LAR) as a variational principle on the information--geometric manifold of amplitudes. In the present section we recast it as an equivalent first--order system on the doubled ambient space \(\mathbb{R}^{2n}\cong\mathbb{R}^n\times\mathbb{R}^n\) of amplitudes and their residuals (momenta), remaining entirely within the classical information--geometric framework. The same equations can be represented as  a hyperbolic Schr\"odinger--type equation in the split--complex algebra \(\mathbb D\) (\autoref{subsec:split-complex}).

\subsection{In $\RR^{2n}$}
\label{subsec:R2n}

To pass from the constrained, pointwise projection dynamics on the unit sphere to an unconstrained system on a linear space, we
encode the LAR least--squares condition by a quadratic \emph{residual} Lagrangian.
Given a preferred ambient direction field $\widehat V\tilde\rho$, define the instantaneous residual
$u:=\dot{\tilde\rho}-\widehat V\tilde\rho$ and penalise its Euclidean norm.
The resulting first--order cotangent lift on $T^*\RR^n\simeq\RR^{2n}$ will be shown below to reproduce the spherical LAR flow for
the normalised amplitudes $\rho=\tilde\rho/\|\tilde\rho\|$ on the distinguished zero--residual sector.

To do so, we define the least--action Lagrangian
\begin{equation}
	\label{eq:least-action-lagrangian-lifted}
L(\tilde\rho,\dot{\tilde\rho})
=\tfrac12\|\dot{\tilde\rho}-\widehat V\tilde\rho\|^2,
\qquad \text{where } \widehat V=\widehat S+\widehat F,\ 
\widehat S^\top=\widehat S,\ 
\widehat F^\top=-\widehat F,
\end{equation}
defined on the ambient space of \emph{unnormalised amplitudes} $\tilde\rho\in\widetilde{\mathcal M}=\RR^n$
equipped with the Euclidean metric.  Throughout Section~\ref{sec:3}, unless otherwise stated, we adopt the trace--free
convention $\tr(\widehat V)=0$ (as in Section~\ref{sec:2}).

The simplex $\mathcal M=\Delta^{n-1}_{>0}$ enters only through the sign--blind readout map
$q(\tilde\rho)$ defined for $\tilde\rho\in\RR^n\setminus\{0\}$ by
\[
q_i=\frac{\tilde\rho_i^2}{\|\tilde\rho\|^2},
\]
and, when expressing simplex--level geometric objects, through the square--root section
$\iota:\mathcal M\to\mathbb S^{n-1}_+$, $\iota(q)=\sqrt q$.
No sign restriction is imposed on the latent amplitude $\tilde\rho$: trajectories may have any sign pattern and may cross
coordinate hyperplanes without affecting the readout. Accordingly, Fisher--Rao geometry (the metric $g_F$ and its derived operators) is invoked only on intervals where the readout
satisfies $q(t)\in\mathcal M$ (and otherwise is interpreted in the stratified sense of Remark \ref{rem:stratified-reading}),
whereas the amplitude/cotangent dynamics are global on $\widetilde{\mathcal M}=\RR^n$ and
$T^*\widetilde{\mathcal M}\simeq\RR^{2n}$.

The Legendre transform of \ref{eq:least-action-lagrangian-lifted} 
\[
\mathbb FL:\;(\tilde\rho,\dot{\tilde\rho})\longmapsto(\tilde\rho,y),
\qquad
y=\partial_{\dot{\tilde\rho}}L=\dot{\tilde\rho}-\widehat V\tilde\rho,
\]
is hyperregular and therefore a diffeomorphism.
Moreover, the Poincaré--Cartan form
$\theta_L=y^\top d\tilde\rho$ satisfies
$\omega_L=-d\theta_L$, and the pullback identity
\[
(\mathbb FL)^*\Omega=\omega_L,
\qquad
\Omega=d\tilde\rho^i\wedge dy_i,
\]
where $\Omega$ is the canonical symplectic $2$--form on $T^*\widetilde{\mathcal M}$,
shows that $\mathbb FL:(T\widetilde{\mathcal M},\omega_L)\to
(T^*\widetilde{\mathcal M},\Omega)$ is a \emph{symplectomorphism}.
Thus the residual variable
\[
u=\dot{\tilde\rho}-\widehat V\tilde\rho
\]
is the canonical momentum $y$ expressed in velocity coordinates,
and $(\tilde\rho,y)$ are Darboux coordinates for the ambient phase space
$T^*\widetilde{\mathcal M}\simeq\RR^{2n}$.

The corresponding Hamiltonian
\begin{equation} \label{eq:hamiltonian-fn}
	H(\tilde\rho,y)
	=\tfrac12\|y\|^2+\langle y,\widehat V\tilde\rho\rangle
\end{equation}
generates the linear Hamiltonian system
\begin{equation}\label{eq:ham-system}
	\dot{\tilde\rho}=\widehat V\tilde\rho+y,
	\qquad
	\dot y=-\,\widehat V^{\!\top}y,
\end{equation}
or in vector notation
\begin{equation} \label{eq:ham-system-vectornotation}
	\dot Z=\mathsf A\, Z,
	\qquad
	Z=\begin{pmatrix}\tilde\rho\\[2pt]y\end{pmatrix},
	\qquad
	\mathsf A=
	\begin{pmatrix}
		\widehat V & I\\[2pt]
		0 & -\widehat V^{\!\top}
	\end{pmatrix},
	\qquad
	\mathsf A^\top J+J\mathsf A=0.
\end{equation}
Hence $\Phi_t=\exp(t\mathsf A)\in\mathrm{Sp}(2n,\RR)$ preserves the canonical symplectic form.

\medskip

The geometric structure underlying this Hamiltonian system arises from the
information--geometric origin of the model.
In information geometry, exponential and mixture families are \emph{dually flat} and hence carry a Hessian structure
\citep{Shima2007,Nielsen2020Entropy}.
In particular, the probability simplex $\mathcal M=\Delta^{n-1}_{>0}$ may be viewed as a mixture family in the coordinates $q$,
so it comes equipped with the flat mixture connection $\nabla^{(m)}$ and the Fisher--Rao metric $g_F$, which is locally Hessian.
Concretely, with the (negative) entropy potential
$
\Phi(q):=\sum_{i=1}^n q_i\log q_i,
$
one has $g_F=(\nabla^{(m)} d\Phi)\!\mid_{T\mathcal M}$.

The square--root lift
\[
\iota:\mathcal M\longrightarrow\mathbb S^{n-1}_+\subset \mathbb S^{n-1}\subset\widetilde{\mathcal M}\simeq\RR^n,
\qquad \iota(q)=\sqrt q,
\]
is a smooth embedding that transports the metric geometry to the ambient amplitude space: it satisfies
$\iota^*g_{\mathbb S}=\tfrac14\,g_F$, where $g_{\mathbb S}=g_{\mathrm{Euc}}|_{\mathbb S^{n-1}}$ is the round metric.
We view $\mathbb S^{n-1}_+$ as a submanifold of the Euclidean space
$(\widetilde{\mathcal M},g_{\mathrm{Euc}})$ and endow $\widetilde{\mathcal M}$ with the standard flat (hence Levi--Civita) connection
$\nabla^{\mathrm{Euc}}$.
Thus $(\widetilde{\mathcal M},g_{\mathrm{Euc}},\nabla^{\mathrm{Euc}})$ provides a canonical ambient affine structure for the
cotangent/Hamiltonian formulation, naturally associated with the underlying statistical geometry of $(\mathcal M,g_F,\nabla^{(m)})$.

\medskip
A crucial geometric consequence of the Hessian structure is that the flat affine connection is part of the given data.
In the present setting this distinguished connection is $\nabla^{(m)}$ on $\mathcal M$.
In general, Sasaki--type lifts on tangent and cotangent bundles depend on a chosen connection; for a Hessian manifold
$(M,g,\nabla)$ the torsion--free flat connection $\nabla$ is fixed in advance.
On the tangent side, this underlies the
classical Sasaki/Dombrowski construction and its Hessian refinement: the
tangent bundle \(TM\) of a dually flat or Hessian manifold carries a canonical
K\"ahler structure induced by \((g,\nabla)\), see
\cite{Dombrowski1962} and Section 2.2 in \cite{Shima2007}, and its
information--geometric applications to quantum mechanics developed by
\cite{Molitor2013JGP}. 
Here, we will instead follow the corresponding approach on the cotangent
side, following \citep{Reginatto2013,Caticha2021}, where the chosen flat
connection provides a distinguished way of lifting curves and covectors:
parallel transport of covectors with respect to the dual connection induced by
$\nabla$ defines, at each point $\alpha\in T^*M$, a horizontal subspace
$H_\alpha\subset T_\alpha(T^*M)$ complementary to the vertical subspace
$V_\alpha:=\ker d\nu_\alpha$, where $\nu:T^*M\to M$ is the bundle projection.
 These subspaces vary smoothly and give a canonical splitting
\[
T(T^*M)=H\oplus V
\]
into \(\nabla\)--horizontal and vertical subbundles.  The next lemma summarise
their basic symplectic properties.

\begin{lemma}[Horizontal--vertical Lagrangian structure on \(T^*M\)]
	\label{lem:HV-bilagrangian}
	Let \(M\) be a smooth manifold equipped with a torsion--free affine
	connection \(\nabla\), and let \(\nu:T^*M\to M\) be the bundle projection.
	Denote by \(\theta\) the tautological one--form on \(T^*M\) and by
	\(\Omega:=-d\theta\) its canonical symplectic form.
	Let \(V:=\ker d\nu\subset T(T^*M)\) be the vertical distribution and let
	\(H\subset T(T^*M)\) be the horizontal distribution determined by \(\nabla\)
	via parallel transport of covectors.
	Then \(V\) and \(H\) are Lagrangian distributions for \((T^*M,\Omega)\),
	they are everywhere transversal, and
	\[
	T(T^*M) = H \oplus V.
	\]
\end{lemma}

\begin{proof}
	In local coordinates \((x^i)\) on \(M\) with corresponding coordinates
	\((x^i,y_i)\) on \(T^*M\), the tautological one--form is
	\(\theta = y_i\,dx^i\), so \(\Omega = -d\theta = -\,dy_i\wedge dx^i = dx^i\wedge dy_i\).
	The vertical distribution \(V=\ker d\nu\) is spanned by
	\(\partial/\partial y_i\), and since \(\Omega=dx^i\wedge dy_i\) one has
	\(\Omega(\partial/\partial y_i,\partial/\partial y_j)=0\), hence \(\Omega|_V = 0\); since \(\dim V = \tfrac12\dim T(T^*M)\),
	\(V\) is Lagrangian.

	The connection \(\nabla\) with Christoffel symbols \(\Gamma^k_{ij}\)
	determines the horizontal lifts
	\[
	\frac{\delta}{\delta x^i}
	:= \frac{\partial}{\partial x^i}
	+ \Gamma^k_{ij}(x)\,y_k\,\frac{\partial}{\partial y_j},
	\]
	which span the horizontal distribution \(H\).  One computes
	\[
	\Omega\Big(\frac{\delta}{\delta x^i},\frac{\delta}{\delta x^j}\Big)
	= (\Gamma^l_{ji}-\Gamma^l_{ij})\,y_l.
	\]
	Since \(\nabla\) is torsion--free, \(\Gamma^l_{ij}=\Gamma^l_{ji}\) and hence
	\(\Omega|_H=0\).  Again \(\dim H = \tfrac12\dim T(T^*M)\), so \(H\) is
	Lagrangian.  By construction \(H\cap V = \{0\}\) and every tangent vector
	splits uniquely into horizontal and vertical components, giving
	\(T(T^*M) = H\oplus V\).
\end{proof}
Lemma~\ref{lem:HV-bilagrangian} provides the symplectic input: for any
torsion--free affine connection $\nabla$ the induced horizontal distribution
$H$ is Lagrangian, complementary to the vertical Lagrangian $V=\ker d\nu$, and
hence determines a canonical Lagrangian splitting
$T(T^*M)=H\oplus V$.
This splitting already defines an \emph{almost} para--Hermitian package by
setting $K|_H=+\id$, $K|_V=-\id$ and $G:=\Omega(\,\cdot\,,K\,\cdot\,)$; however,
to obtain a \emph{para--K\"ahler} structure in the strict sense one must also
have integrability of $H$ (equivalently, of $K$), which is ensured by flatness
of $\nabla$.  We therefore specialise to the Hessian case:
\begin{proposition}[Cotangent para--K\"ahler structure of a Hessian manifold]
	\label{thm:cotangent-parakahler}
	Let \((M,g,\nabla)\) be a Hessian manifold, so that \(\nabla\) is torsion--free
	and flat.  Then the cotangent bundle \(T^*M\) carries a canonical
	para--K\"ahler structure \((G,K,\Omega)\), where
	\begin{itemize}
		\item \(\Omega\) is the canonical symplectic form on \(T^*M\),
		\item \(K\) is the para--complex structure
		\(K:T(T^*M)\to T(T^*M)\) induced by the
		\(\nabla\)--horizontal/vertical splitting
		\(T(T^*M)=H\oplus V\), and
		\item \(G\) is the unique neutral metric satisfying
		\(\Omega(u,v)=G(Ku,v)\) for all \(u,v\in T(T^*M)\).
	\end{itemize}
	The \(\pm1\) eigensubbundles \(E_\pm=\ker(K\mp\mathrm{id})\) coincide with
	\(H\) and \(V\), are maximally isotropic for \(G\), and yield the canonical
	Lagrangian splitting
	\[
	T(T^*M)=E_+\oplus E_-.
	\]
\end{proposition}

\begin{proof}
	Since $(M,g,\nabla)$ is Hessian, the affine connection $\nabla$ is
	torsion--free and flat. Let $\nu:T^*M\to M$ be the projection, $\theta$ the
	tautological one--form and $\Omega:=-d\theta$ the canonical symplectic form on
	$T^*M$.  Denote by $V:=\ker d\nu\subset T(T^*M)$ the vertical distribution
	and by $H\subset T(T^*M)$ the $\nabla$--horizontal distribution obtained by
	parallel transport of covectors. By Lemma~\ref{lem:HV-bilagrangian}, $H$ and
	$V$ are transversal Lagrangian distributions for $(T^*M,\Omega)$ and
	$T(T^*M)=H\oplus V$.  Flatness of $\nabla$ implies that $H$ is integrable,
	while $V$ is always integrable, so $H$ and $V$ define transverse Lagrangian
	foliations and $(T^*M,\Omega;H,V)$ is a bi--Lagrangian manifold.
	
	By \cite[Theorem~6 and Proposition~7(b)]{Etayo2006}, any
	bi--Lagrangian manifold $(N,\omega;F_1,F_2)$ admits a canonical neutral metric
	$G$ and an almost product (para--complex) structure $K$ with $K^2=\id$, whose
	$\pm1$ eigendistributions are the tangent bundles of $F_1$ and $F_2$, and
	such that $G(u,v)=\omega(Ku,v)$ for all $u,v$.  Applying this with
	$N=T^*M$, $\omega=\Omega$, $F_1=H$, $F_2=V$ yields a neutral metric $G$ and
	an almost para--complex structure $K$ on $T^*M$ with
	\[
	K|_H=\id,\qquad K|_V=-\id,\qquad
	\Omega(u,v)=G(Ku,v)\quad\forall u,v\in T(T^*M).
	\]
	In particular, the $\pm1$ eigensubbundles
	\[
	E_+ := \ker(K-\id)=H,\qquad E_- := \ker(K+\id)=V
	\]
	are maximally isotropic for $G$ and give the Lagrangian splitting
	$T(T^*M)=E_+\oplus E_-$.  The uniqueness of $G$ with
	\(\Omega(u,v)=G(Ku,v)\) is also part of \cite[Theorem~6]{Etayo2006}.
\end{proof}

Lemma~\ref{lem:HV-bilagrangian} requires only torsion--freeness to obtain a
transversal Lagrangian splitting $T(T^*M)=H\oplus V$ and hence an \emph{almost}
para--Hermitian package $(G,K,\Omega)$.  The additional ``Hessian'' hypothesis
(flatness of $\nabla$) is used only to guarantee integrability of $H$ (and
therefore of $K$), so that the structure is para--K\"ahler in the strict sense.
Then Proposition~\ref{thm:cotangent-parakahler} shows that, once a torsion--free flat
connection \(\nabla\) is fixed (in particular, when \((M,g,\nabla)\) is a
Hessian manifold), the cotangent bundle \(T^*M\) carries a \emph{canonical}
para--K\"ahler structure \((G,K,\Omega)\).  The construction depends only on
the canonical symplectic form \(\Omega=-d\theta\) on \(T^*M\) and on the
\(\nabla\)--horizontal/vertical splitting \(T(T^*M)=H\oplus V\). In this sense the 
para--complex structure \(K\) is intrinsically determined by this splitting,
and \(G\) is then uniquely determined by the relation \(\Omega(u,v)=G(Ku,v)\).

By contrast, there is no complex structure \(J\) on \(T^*M\) that is canonically
determined by the affine/Hessian data \((M,g,\nabla)\) alone and compatible with
the canonical symplectic form \(\Omega\).
Of course, once one chooses an \(\Omega\)--compatible integrable almost complex
structure \(J\) on \(T^*M\), the associated metric \(g_J(\cdot,\cdot):=\Omega(\cdot,J\cdot)\)
makes \((T^*M,g_J,J,\Omega)\) K\"ahler.  The point is that \(J\) is extra
structure: it is not canonically selected by \(\nabla\) or by the Hessian metric
on \(M\), and different choices of \(J\) lead to genuinely different K\"ahler
lifts.  In this sense, the para--K\"ahler structure of
Proposition~\ref{thm:cotangent-parakahler} is canonical relative to the Hessian
data, whereas an elliptic K\"ahler structure on \(T^*M\) is not.
We return to this in \autoref{subsec:non-hermitian-packaging} and \autoref{sec:4}.

Next, specialising Proposition~\ref{thm:cotangent-parakahler} to the lifted Hessian
manifold $(\widetilde{\mathcal M},g_{\mathrm{Euc}},\nabla^{\mathrm{Euc}})$
yields a canonical para--K\"ahler structure $(G,K,\Omega)$ on
$T^*\widetilde{\mathcal M}\simeq\RR^{2n}$.  The least--action Hamiltonian
system~\eqref{eq:ham-system} is Hamiltonian; hence its flow preserves the
canonical symplectic form $\Omega$.  In particular, in the present linear
setting on $\RR^{2n}$ its generator $\mathsf A$ is an element of
$\mathfrak{sp}(2n,\RR)$.

\begin{corollary}[Canonical \(K\)--splitting and neutral form in the Euclidean Hessian lift]
	\label{cor:ambient-parasplitting-geom}
	Let \((G,K,\Omega)\) be the canonical para--K\"ahler structure on \(T^*\widetilde{\mathcal M}\)
	induced by the Hessian data \((\widetilde{\mathcal M},g_{\mathrm{Euc}},\nabla^{\mathrm{Euc}})\)
	as in Proposition~\ref{thm:cotangent-parakahler}.  Then the \(\pm1\) eigendistributions of \(K\)
	coincide with the horizontal/vertical splitting:
	\[
	E_+=H,\qquad E_-=V,\qquad K|_{E_+}=+\id,\ \ K|_{E_-}=-\id,
	\]
	where \(H\) is the \(\nabla^{\mathrm{Euc}}\)--horizontal distribution and \(V=\ker d\nu\) is the vertical
	distribution for the bundle projection \(\nu:T^*\widetilde{\mathcal M}\to\widetilde{\mathcal M}\).
	In the global Darboux chart \((\tilde\rho,y)\) on \(T^*\widetilde{\mathcal M}\simeq\RR^{2n}\) one has
	\[
	H=\mathrm{span}\Bigl\{\frac{\partial}{\partial\tilde\rho^i}\Bigr\},
	\qquad
	V=\mathrm{span}\Bigl\{\frac{\partial}{\partial y_i}\Bigr\},
	\]
	and
	\[
	\Omega=d\tilde\rho^i\wedge dy_i,
	\qquad
	G=\sum_{i=1}^n\Bigl(d\tilde\rho^i\otimes dy_i+dy_i\otimes d\tilde\rho^i\Bigr).
	\]
	The associated neutral quadratic form \(Q:T^*\widetilde{\mathcal M}\to\RR\) is
	\begin{equation}\label{eq:neutral-Q}
		Q(\tilde\rho,y):=G_{(\tilde\rho,y)}\bigl((\tilde\rho,y),(\tilde\rho,y)\bigr)
		\quad\text{(using the translation identification }T_{(\tilde\rho,y)}\RR^{2n}\simeq\RR^{2n}\text{)}
		=2\langle \tilde\rho,y\rangle.
	\end{equation}
\end{corollary}

\begin{proof}
	Since $\nabla^{\mathrm{Euc}}$ is the standard flat torsion--free connection on
	$\widetilde{\mathcal M}\subset\RR^n$, its Christoffel symbols vanish in the
	global chart $(\tilde\rho^i)$.  Hence the horizontal lifts are
	$\delta/\delta\tilde\rho^i=\partial/\partial\tilde\rho^i$, giving
	$H=\mathrm{span}\{\partial/\partial\tilde\rho^i\}$, while
	$V=\ker d\nu=\mathrm{span}\{\partial/\partial y_i\}$.  Proposition~\ref{thm:cotangent-parakahler}
	defines $K$ by $K|_H=+\id$ and $K|_V=-\id$, so $E_+=H$ and $E_-=V$.
	
	In the Darboux chart, $\theta=y_i\,d\tilde\rho^i$ and hence
	$\Omega=-d\theta=d\tilde\rho^i\wedge dy_i$.  The metric $G$ is determined by
	$\Omega(u,v)=G(Ku,v)$; with $K|_H=+\id$ and $K|_V=-\id$ this yields
	$G=\sum_i(d\tilde\rho^i\otimes dy_i+dy_i\otimes d\tilde\rho^i)$.
	Finally, $Q(\tilde\rho,y)=G((\tilde\rho,y),(\tilde\rho,y))=2\langle\tilde\rho,y\rangle$
	(in the translation identification $T_{(\tilde\rho,y)}\RR^{2n}\simeq\RR^{2n}$),
	and $\Lambda=Q\circ z$ by definition.
\end{proof}

\medskip
\noindent
With the canonical $K$--splitting $T(T^*\widetilde{\mathcal M})=E_+\oplus E_-$ from
Corollary~\ref{cor:ambient-parasplitting-geom} at hand, we rewrite the least--action
Hamiltonian system in the corresponding $K$--adapted (Witt) coordinates $(z_+,z_-)=(\tilde\rho,y)$ on
$T^*\widetilde{\mathcal M}\simeq\RR^{2n}$, where (under the translation identification) the distributions
$E_+$ and $E_-$ are spanned by $\partial/\partial z_+^i$ and $\partial/\partial z_-{}_i$, respectively.
This exposes the dynamics as a
block flow perturbed by a shear mixing $E_-$ into $E_+$, and it yields the
monotonicity identity for the neutral coordinate $\Lambda(t)=2\langle z_+(t),z_-(t)\rangle$:

\begin{proposition}[Witt form and shear decomposition of the least--action lift]
	\label{prop:shear-form}
	In the $K$--adapted coordinates $(z_+,z_-)=(\tilde\rho,y)$ of
	Corollary~\ref{cor:ambient-parasplitting-geom}, the least--action Hamiltonian system
	\eqref{eq:ham-system} takes the Witt form
	\begin{equation}\label{eq:eom-witt}
		\dot z_+=\widehat V z_+ + z_-,
		\qquad
		\dot z_-=-\widehat V^{\!\top}z_-.
	\end{equation}
	Equivalently, its generator decomposes as
	\begin{equation}\label{eq:A-decomp-pu-shear}
		\mathsf A=
		\underbrace{\begin{pmatrix}\widehat V&0\\[2pt]0&-\widehat V^{\!\top}\end{pmatrix}}_{\mathsf A_{\mathrm{pu}}}
		+
		\underbrace{\begin{pmatrix}0&I\\[2pt]0&0\end{pmatrix}}_{\mathsf A_{\mathrm{sh}}},
	\end{equation}
	where $\mathsf A_{\mathrm{pu}}$ commutes with $K$ (hence preserves $E_\pm$) and
	$\mathsf A_{\mathrm{sh}}$ is a nilpotent shear ($\mathsf A_{\mathrm{sh}}^2=0$) mapping
	$E_-$ into $E_+$.  In particular, the zero--residual leaf $\{z_-=0\}$ is invariant.
	
	Moreover, the neutral coordinate $\Lambda(t)=2\langle z_+(t),z_-(t)\rangle$ satisfies
	\begin{equation}\label{eq:Lambda-dot-here}
		\dot\Lambda(t)=2\|z_-(t)\|^2=2\|y(t)\|^2\ \ge\ 0.
	\end{equation}
\end{proposition}
\begin{proof}[Proof of Proposition~\ref{prop:shear-form}]
	We start with equation \eqref{eq:ham-system-vectornotation}.	
	With the \(K\)--adapted splitting \(TN=E_+\oplus E_-\) and the corresponding
	coordinates \(z_+=\tilde\rho\), \(z_-=y\) (cf.\ Corollary~\ref{cor:ambient-parasplitting-geom}),
	this is precisely \eqref{eq:eom-witt}.  Next, decompose \(\mathsf A\) into its
	block--diagonal and strictly upper--triangular parts,
	\[
	\mathsf A=\mathsf A_{\mathrm{pu}}+\mathsf A_{\mathrm{sh}},
	\qquad
	\mathsf A_{\mathrm{pu}}=
	\begin{pmatrix}
		\widehat V & 0\\[2pt]
		0 & -\widehat V^{\!\top}
	\end{pmatrix},
	\qquad
	\mathsf A_{\mathrm{sh}}=
	\begin{pmatrix}
		0 & I\\[2pt]
		0 & 0
	\end{pmatrix}.
	\]
	This is the canonical block decomposition of \(\mathsf A\).  The shear part
	\(\mathsf A_{\mathrm{sh}}\) is nilpotent (\(\mathsf A_{\mathrm{sh}}^2=0\)) and
	satisfies \(\mathsf A_{\mathrm{sh}}(E_-)\subset E_+\) and
	\(\mathsf A_{\mathrm{sh}}(E_+)=\{0\}\), i.e.\ it maps \(E_-\) into \(E_+\).
	Moreover, in these \(K\)--adapted coordinates one has \(K=\mathrm{diag}(I,-I)\) on
	\(E_+\oplus E_-\).  Since \(\mathsf A_{\mathrm{pu}}\) is block--diagonal, it
	commutes with \(K\), hence preserves \(E_+\) and \(E_-\).
	Invariance of the zero--residual leaf follows directly from \eqref{eq:eom-witt}:
	if \(z_-(t_0)=0\) then \(\dot z_-(t_0)=-\widehat V^{\!\top}z_-(t_0)=0\), hence
	\(z_-(t)\equiv 0\).
	
	Finally, for \(\Lambda(t):=2\langle z_+(t),z_-(t)\rangle=2\langle \tilde\rho(t),y(t)\rangle\),
	differentiate and use \eqref{eq:ham-system}:
	\[
	\dot\Lambda
	=2\bigl(\langle \dot{\tilde\rho},y\rangle+\langle \tilde\rho,\dot y\rangle\bigr)
	=2\bigl(\langle \widehat V\tilde\rho+y,\,y\rangle+\langle \tilde\rho,\,-\widehat V^{\!\top}y\rangle\bigr).
	\]
	The mixed terms cancel because
	\(\langle \widehat V\tilde\rho,\,y\rangle=\langle \tilde\rho,\,\widehat V^{\!\top}y\rangle\),
	leaving \(\dot\Lambda=2\|y\|^2\), which is \eqref{eq:Lambda-dot-here}.
\end{proof}

\begin{remark}[On real polarisation]\label{remark:polarization}
	The para--K\"ahler structure $(G,K,\Omega)$ on $N=T^*\widetilde{\mathcal M}$
	canonically determines a bi--Lagrangian splitting
	\[
	TN=E_+\oplus E_-,
	\qquad
	E_\pm=\ker(K\mp\id),
	\]
	so in the hyperbolic (para--)theory there is a distinguished \emph{real}
	polarisation built into the geometry, namely the integrable Lagrangian
	distribution $E_+$ (equivalently, the Lagrangian foliation by $E_+$--leaves).
	In the global Darboux chart $z=(\tilde\rho,y)$ on $N$, the zero--residual set
	\[
	\Sigma_0:=\{(\tilde\rho,y):y=0\}
	\]
	is the $E_+$--leaf through the origin (the zero section).  In particular,
	$\Sigma_0$ is Lagrangian and $T\Sigma_0=E_+|_{\Sigma_0}$.  Moreover, $\Sigma_0$
	is invariant under the lifted flow, since $\dot y=-\widehat V^{\!\top}y$
	implies $y(t)\equiv 0$ whenever $y(t_0)=0$.  On $\Sigma_0$ the dynamics closes
	as $\dot{\tilde\rho}=\widehat V\tilde\rho$, and the induced simplex process is
	read out by normalisation,
	\[
	q(t)=\frac{\tilde\rho(t)^{\odot2}}{\|\tilde\rho(t)\|^2}\qquad(\tilde\rho(t)\neq 0).
	\]
	We return to this point when making contact with the (para--)geometric
	quantisation template, and contrast it with the genuinely \emph{additional}
	polarisation data required in the elliptic (K\"ahler) setting of
	\autoref{sec:4}.
\end{remark}

In the following, we use the terminology \emph{on/off-shell} relative to the invariant zero section
$
\Sigma_0=\{y=0\}.
$
A trajectory is \emph{on-shell} if it is entirely contained in \(\Sigma_0\), equivalently if \(y(t)\equiv 0\) on the interval under consideration.  It is \emph{off-shell} if it is not confined to \(\Sigma_0\), equivalently if \(y(t)\neq 0\) at least at some time.  Thus off-shell trajectories are not, in general, confined to a fixed leaf of the \(E_+\)-foliation, since the residual variable \(y\) typically evolves.

We are now in the position where we can define a para--qudit model on the zero--residual leaf:
\begin{corollary}[The para--qudit model]
	\label{cor:onshell-simplex-projection}
	In the coordinates $(z_+,z_-)=(\tilde\rho,y)$ of Proposition~\ref{prop:shear-form}, the zero--residual leaf
	\[
	\Sigma_0:=\{(\tilde\rho,y):\ y=0\}
	\]
	is invariant. On $\Sigma_0$ the lifted dynamics reduce to the linear on-zero section amplitude equation
	\begin{equation}\label{eq:cor-amps-onshell}
		\dot{\tilde\rho}(t)=\widehat V\,\tilde\rho(t),
		\qquad
		\widehat V=\widehat S+\widehat F,
		\qquad
		\widehat S^{\!\top}=\widehat S,\quad \widehat F^{\!\top}=-\widehat F,
	\end{equation}
	and for any nonzero initial condition $\tilde\rho(0)\neq 0$ one has
	\begin{equation}\label{eq:cor-tilderho-exp}
		\tilde\rho(t)=e^{t\widehat V}\,\tilde\rho(0).
	\end{equation}
	Define the normaliser
	\begin{equation}\label{eq:cor-Z-def}
		\mathcal Z(t):=\|\tilde\rho(t)\|^{2}=\sum_{i=1}^n \tilde\rho_i(t)^2,
	\end{equation}
	and the simplex readout
	\begin{equation}\label{eq:cor-q-def}
		q_i(t):=\frac{\tilde\rho_i(t)^2}{\mathcal Z(t)}\quad(\tilde\rho(t)\neq 0),
		\qquad
		q(t)\in\Delta^{n-1}.
	\end{equation}
	Moreover, $q(t)\in\mathcal M=\Delta^{n-1}_{>0}$ on any interval on which $\tilde\rho_i(t)\neq 0$ for all $i$.
\end{corollary}

\begin{proof}
	Invariance of $\Sigma_0$ is immediate from $\dot y=-\widehat V^{\!\top}y$ in \eqref{eq:eom-witt}.  On $\Sigma_0$ one has
	$\dot{\tilde\rho}=\widehat V\tilde\rho$, hence $\tilde\rho(t)=e^{t\widehat V}\tilde\rho(0)$.
	If $\tilde\rho(0)\neq 0$ then $\tilde\rho(t)\neq 0$ for all finite $t$ since $e^{t\widehat V}$ is invertible, so
	$\mathcal Z(t)=\|\tilde\rho(t)\|^2>0$.  Therefore $q_i(t)=\tilde\rho_i(t)^2/\mathcal Z(t)\ge 0$ and $\sum_i q_i(t)=1$,
	hence $q(t)\in\Delta^{n-1}$ for all finite $t$.
	
	We use the closed simplex $\Delta^{n-1}$ because the readout map
	\[
	\tilde\rho\ \longmapsto\ q=\frac{\tilde\rho^{\odot2}}{\|\tilde\rho\|^2}
	\]
	is defined for all $\tilde\rho\neq 0$ and may attain boundary points when some coordinate $\tilde\rho_i(t)$ vanishes.
	Such boundary instants do not affect the well-posedness of the on--shell amplitude dynamics
	$\dot{\tilde\rho}=\widehat V\tilde\rho$ nor the positivity of $\mathcal Z(t)$.
\end{proof}

Next we can identify rational inattention (RI) and Projective Expected Utilities (PEU) as special cases in the following two examples:

\begin{example}[Standard ER/RI as a diagonal LAR case]
	\label{example:RI-ER}
	Let the generator be purely symmetric,
	\(
	\widehat V=\widehat S\in\mathrm{Sym}(\mathbb R^n)
	\)
	with \(\widehat F=0\).
	The amplitude dynamics
	\[
	\dot{\tilde\rho}_t=\widehat S\,\tilde\rho_t
	\quad\Rightarrow\quad
	\tilde\rho_T=e^{T\widehat S}\tilde\rho_0
	\]
	project to the simplex as
	\begin{equation}\label{eq:choice-probs}
		q_k(T)
		=\frac{\tilde\rho_{T,k}^2}{\sum_j \tilde\rho_{T,j}^2}
		=\frac{\bigl[(e^{T\widehat S}\tilde\rho_0)_k\bigr]^2}
		{\tilde\rho_0^\top e^{2T\widehat S}\tilde\rho_0},
		\qquad (\widehat S=\widehat S^\top),
	\end{equation}
	where \(q_k(0)=\tilde\rho_{0,k}^2/\sum_j \tilde\rho_{0,j}^2\) is the prior.
	For a general
	(non--diagonal) \(\widehat S\) this yields a coupled exponential tilt
	on the simplex determined by the full operator \(\widehat S\).
	
	In the diagonal (choice–separable) case
	\(
	\widehat S=\mathrm{diag}(\theta_1,\dots,\theta_n)
	\)
	we obtain
	\[
	\tilde\rho_{T,k}=e^{T\theta_k}\tilde\rho_{0,k},
	\qquad
	q_k(T)
	=\frac{q_k(0)\,e^{2T\theta_k}}
	{\sum_j q_j(0)\,e^{2T\theta_j}}.
	\]
	In standard entropy–regularised / RI choice, the posterior has the
	form
	\[
	q_k(\beta)
	=\frac{q_k(0)\,e^{\beta v_k}}{\sum_j q_j(0)\,e^{\beta v_j}},
	\]
	as the maximiser of a linear–plus–entropy objective
	(for an information--theoretic overview, see \cite{Harre2021} and \cite{Buckley2017}). 
	Identifying
	\(v_k=\theta_k\) and \(\beta=2T\), diagonal LAR reproduces the
	standard ER/RI posterior.  In information–geometric terms, the
	corresponding continuous--time flow on the simplex (on $\Delta^{n-1}_{>0}$) is the
	Fisher--Rao natural gradient flow of the linear functional
	\(q\mapsto\sum_k \theta_k q_k\) in the sense of
	\citet{AmariNagaoka2000}.  In the diagonal case this natural gradient
	coincides (up to a constant time rescaling) with the standard
	\textit{replicator dynamics} with payoffs \(\theta_k\)
	\citep{HarperInfoGeomEGT,SandholmEGT}.
\end{example}

\begin{example}[PEU and EUT]
	\label{example:PEU-EUT}
	Projective Expected Utility (PEU) \citep{LaMura2009} represents lotteries by rays
	$[\psi]\subset\RR^n$ and evaluates them by a homogeneous quadratic form
	\[
	U_{\mathrm{PEU}}(\psi)=\psi^\top W\psi,
	\qquad W=W^\top,
	\qquad \psi\sim\lambda\psi\ (\lambda>0).
	\]
	Assume the purely evaluative LAR dynamics of Example~\ref{example:RI-ER},
	\[
	\widehat V=\widehat S\in\Sym(\RR^n),\quad \widehat F=0,\qquad
	\tilde\rho_T=e^{T\widehat S}\tilde\rho_0,\qquad
	q_k(T)=\frac{\tilde\rho_{T,k}^2}{\sum_j \tilde\rho_{T,j}^2}.
	\]
	Assume the largest eigenvalue $\lambda^\star$ of $\widehat S$ is simple, with corresponding eigenvector
	$w^\star\neq 0$ satisfying $\widehat S w^\star=\lambda^\star w^\star$, and assume $\ip{w^\star}{\tilde\rho_0}\neq0$.
	Then the amplitude ray converges projectively,
	\[
	[\tilde\rho_T]\to[w^\star]\quad\text{as }T\to\infty,
	\qquad\text{hence}\qquad
	q(T)\to q^\star,\ \ q^\star_k:=\frac{(w^\star)_k^2}{\sum_j (w^\star)_j^2},
	\]
	where $(w^\star)_k$ denotes the $k$th component of $w^\star$ in the choice (standard) basis.
	Thus the large--$T$ (max--plus / tropical) limit of symmetric LAR selects a canonical projective state $[w^\star]$
	and an induced limiting lottery $q^\star$ on the simplex, providing the natural PEU object associated with $\widehat S$.
	
	In the diagonal (choice--separable) subcase $\widehat S=\diag(\theta_1,\dots,\theta_n)$ with unique maximiser
	$\theta_{k^\star}=\max_k\theta_k$, we have $w^\star=e_{k^\star}$ and therefore $q(T)\to\delta_{k^\star}$,
	recovering deterministic expected utility theory (EUT) choice.
	Moreover, if $W=\diag(u_k)$ and $\mathcal Z(T):=\sum_j\tilde\rho_{T,j}^2$, then
	\[
	U_{\mathrm{PEU}}(\tilde\rho_T)=\sum_k u_k\tilde\rho_{T,k}^2
	=\mathcal Z(T)\sum_k u_k q_k(T),
	\]
	so PEU evaluates the amplitude ray $[\tilde\rho_T]$ while EUT evaluates the expected utility $\sum_k u_k q_k(T)$
	of the induced lottery $q(T)$ on the simplex (they coincide up to the homogeneous scale factor $Z(T)$ in the diagonal case).
\end{example}

For a symmetric but non--diagonal $\widehat S$, the map $q(T)$ in \eqref{eq:choice-probs} is still given by the same
amplitude evolution $\tilde\rho_T=e^{T\widehat S}\tilde\rho_0$, but now depends on the full operator $\widehat S$
(and hence on all components of $\tilde\rho_0$ in the choice basis).  In particular, in the choice basis this
generically cannot be reduced to a coordinatewise exponential tilt
\[
q_k(T)\propto q_k(0)\exp\{\eta v_k\}
\qquad\text{(for full--support $q(0)\in\mathcal M$).}
\]
for any single vector $v\in\RR^n$.  Equivalently, the diagonal ER/RI form is recovered if and only if $\widehat S$
is simultaneously diagonal in the choice basis; otherwise the off--diagonal entries mix coordinates and act as a
rotation of the evaluative basis, inducing genuine \emph{contextual} couplings between alternatives.  Thus, within
the same LAR model, the diagonal (choice--separable) regime reproduces classical entropy--regularised / RI posteriors,
whereas non--diagonal symmetric generators produce the contextual, interference--like, and order--sensitive
distortions of choice probabilities characteristic of quantum--like (QL) models.  In this sense ER/RI and QL appear
as two regimes of a single information--geometric dynamical framework.

To summarise, we have now obtained the dynamical equation \eqref{eq:eom-witt} in the real doubled space
\(\RR^{2n}\).  This real formulation is the most direct, since it keeps the linear--algebraic dual spaces explicit and makes the canonical split metric transparent, but it is often advantageous to repackage the same system in different algebraic languages. 
Because the doubled space carries a natural split signature, it is convenient to 
introduce the split--complex \(\mathbb D^n\) representation in \autoref{subsec:split-complex}, deriving a (split) para--Schr\"odinger equation.  
This will also allow a clean comparison, in \autoref{subsec:non-hermitian-packaging}, with the 
standard complex packaging in \(\mathbb{C}^n\).

\subsection{In $\mathbb D$: the hyperbolic Schr\"odinger equation}
\label{subsec:split-complex}

Having derived the lifted least--action equations of motion on
\(N=T^*\widetilde{\mathcal M}\simeq\RR^{2n}\) (cf.\ \eqref{eq:ham-system}), we now show how the
shear--free dynamics on the zero--residual leaf admit an intrinsic split--complex
packaging.  In this representation the canonical Witt splitting is encoded by
the idempotents of \(\mathbb D\), and the on--shell flow takes the compact
hyperbolic Schr\"odinger form.

We stay in the global Darboux chart \((\tilde\rho,y)\in\RR^n\oplus\RR^n\) on
\(N=T^*\widetilde{\mathcal M}\), with canonical symplectic form
\[
\Omega=d\tilde\rho^i\wedge dy_i.
\]
By Proposition~\ref{prop:shear-form}, the flat Euclidean para--K\"ahler structure determines a para--complex
endomorphism \(K\) with Witt decomposition \(T(T^*\widetilde{\mathcal M})=E_+\oplus E_-\),
\(E_\pm=\ker(K\mp\id)\), and the lifted least--action flow is linear with generator
\begin{equation}\label{eq:A-split}
	\dot Z=\mathsf A Z,
	\qquad
	\mathsf A=\mathsf A_{\mathrm{pu}}+\mathsf A_{\mathrm{sh}},
\end{equation}
where \(\mathsf A_{\mathrm{pu}}\) preserves the Witt splitting (it commutes with \(K\)) and
\(\mathsf A_{\mathrm{sh}}\) is a nilpotent shear mapping \(E_-\) into \(E_+\).
In the \(K\)--adapted coordinates \((z_+,z_-)=(\tilde\rho,y)\), the zero--residual leaf
\(\Sigma_0=\{z_-=0\}\) is invariant, and on \(\Sigma_0\) the dynamics reduce to the decoupled Witt evolution
\begin{equation}\label{eq:witt-onshell}
	\dot z_\pm = (\widehat F\pm \widehat S)\,z_\pm,
	\qquad\text{with } z_-\equiv 0,
\end{equation}
which is the regime naturally encoded by the split--complex (hyperbolic) Schr\"odinger form.

\medskip

To see this and to package \eqref{eq:witt-onshell} intrinsically, we use the split--complex algebra
\(\mathbb D=\{a+bj\mid a,b\in\RR,\ j^2=+1\}\) and its idempotents
\(e_\pm=\tfrac12(1\pm j)\), which satisfy \(j e_\pm=\pm e_\pm\) and provide a complete pair of orthogonal
projectors.\footnote{Terminology varies by field.
	The scalar algebra generated by a unit \(j\) with \(j^{2}=+1\) is known as the \emph{split--complex},
	\emph{hyperbolic}, or \emph{double} numbers. In differential geometry, the corresponding endomorphism \(K\)
	with \(K^{2}=\id\) and equal--rank eigenspaces is a \emph{para--complex} structure; together with a neutral
	metric \(G\) and compatible symplectic form \(\Omega\), the triple \((G,K,\Omega)\) is \emph{para--K\"ahler}.}
These algebraic projectors mirror the Witt projectors \(P_\pm=\tfrac12(\id\pm K)\) onto \(E_\pm\), and thus
identify the real direct sum \(E_+\oplus E_-\cong\RR^n\oplus\RR^n\) with the free \(\mathbb D\)--module
\(\mathbb D^n\).  Accordingly, any Witt state \((z_+,z_-)\) can be encoded as the split--complex state
\[
\Psi_{\mathbb D} := z_+ e_+ + z_- e_- \in \mathbb D^n.
\]
Viewed as a real vector space, \(\RR^{2n}\) equipped with the neutral metric \(G\) is a Krein space; the
split--complex packaging simply recasts the same para--Hermitian geometry in module form.

\medskip

On the zero--residual leaf, it is convenient to introduce the para--Hermitian Hamiltonian operator
\[
\widehat H_{\mathbb D} = \widehat S + j\,\widehat F \in \End(\mathbb D^n),
\qquad
\widehat H_{\mathbb D}^{\#} = \widehat H_{\mathbb D},
\]
where \((\cdot)^\#\) denotes split--conjugate transpose, i.e.\ \((a+bj)^\#:=a-bj\) and
\(\widehat A^\#:=(\widehat A^{\#\textnormal{(entrywise)}})^{\!\top}\).
The two real equations in \eqref{eq:witt-onshell} then combine into the compact para--Schr\"odinger equation
\begin{equation}\label{eq:para-schrodinger}
	\dot\Psi_{\mathbb D} = j\,\widehat H_{\mathbb D}\,\Psi_{\mathbb D},
	\qquad
	\widehat H_{\mathbb D}^{\#}=\widehat H_{\mathbb D}.
\end{equation}
Projecting \eqref{eq:para-schrodinger} onto \(e_\pm\) recovers the Witt components, since \(j e_\pm=\pm e_\pm\)
implies \(\dot z_\pm=(\widehat F\pm\widehat S)z_\pm\). Away from the zero--residual leaf, the full lifted dynamics \eqref{eq:ham-system}
can be viewed as the para--Schr\"odinger evolution supplemented by the nilpotent
shear term \(\mathsf A_{\mathrm{sh}}\), i.e.\ a forcing from \(E_-\) into \(E_+\)
that vanishes exactly on \(\Sigma_0\).
In this formulation \(\widehat H_{\mathbb D}\) is para--Hermitian, \((\widehat H_{\mathbb D})^\#=\widehat H_{\mathbb D}\), hence
\((j\widehat H_{\mathbb D})^\#=-\,j\widehat H_{\mathbb D}\).  Therefore the propagator
\[
U(t):=\exp(j\widehat H_{\mathbb D}\,t)
\]
is para--unitary, \(U(t)^\#U(t)=I\), and preserves the induced para--Hermitian (Krein) pairing
\(\langle\Psi,\Phi\rangle_{\mathbb D}:=\Psi^\#\Phi\) (equivalently its real part \(G=\Re\langle\cdot,\cdot\rangle_{\mathbb D}\))
along the shear--free on--shell flow on \(\Sigma_0\).  The split--complex packaging thus provides a compact description in
which the hyperbolic Schr\"odinger form holds precisely on the zero--residual leaf.

\subsection{In $\mathbb C$: non--Hermitian packaging}
\label{subsec:non-hermitian-packaging}

The split--complex representation in \autoref{subsec:split-complex} is intrinsic to the Hessian
para--K\"ahler lift: the para--complex structure \(K\) and its Witt splitting are part of the
geometric data of the model, and the shear--free flow on the zero--residual leaf is naturally
expressed as a para--Schr\"odinger equation.  One may nevertheless \emph{repackage the same real lifted
	equations} in complex coordinates by choosing an \emph{auxiliary} almost complex structure
\(J\) with \(J^2=-\id\).  Such choices always exist (e.g.\ \(\Omega\)--compatible ones on a symplectic
manifold), but they are noncanonical and affect only how one represents the fixed symplectic form
\(\Omega\) (for instance via \(\Omega(\cdot,J\cdot)\)), not \(\Omega\) itself.
Consequently, the resulting complex form of the equations need not be unitary and, in general, does not
reflect the underlying para--Hermitian geometry.

To make this concrete, in the flat Darboux chart \((\tilde\rho,y)\in\RR^n\times\RR^n\) on
\(\RR^{2n}\simeq T^*\widetilde{\mathcal M}\) we choose the standard almost complex structure that pairs
base and fibre coordinates.  Define
\[
J(\tilde\rho,y)=(-y,\tilde\rho),
\qquad J^2=-\id,
\]
and identify \(\RR^{2n}\) with \(\mathbb C^n\) via
\begin{equation}
	\label{eq:complex-identification}
	z := \tilde\rho + i\,y,
	\qquad 
	\tilde\rho = \tfrac12(z+\bar z),
	\qquad 
	y = \tfrac{1}{2i}(z-\bar z).
\end{equation}
Under this identification, the real lifted Hamiltonian system \eqref{eq:ham-system},
\[
\dot{\tilde\rho}=\widehat V\,\tilde\rho+y,
\qquad 
\dot y=-\,\widehat V^{\!\top}y,
\qquad 
\widehat V=\widehat S+\widehat F,\ \
\widehat S^{\!\top}=\widehat S,\ \
\widehat F^{\!\top}=-\widehat F,
\]
becomes, after substituting \eqref{eq:complex-identification} and collecting the \(z\) and \(\bar z\) components,
\begin{equation}
	\label{eq:eom-non-hermitian}
	\dot z=\Bigl(\widehat F-\tfrac{i}{2}I\Bigr)z
	\;+\;
	\Bigl(\widehat S+\tfrac{i}{2}I\Bigr)\bar z.
\end{equation}
Equivalently,
\[
\dot z
=
\widehat F z + \widehat S \bar z
\;-\;\frac{i}{2}(z-\bar z)
=
\widehat F z + \widehat S \bar z + y,
\qquad
y=\frac{1}{2i}(z-\bar z),
\]
so the scalar terms \(\pm\tfrac{i}{2}I\) in \eqref{eq:eom-non-hermitian} are simply the canonical \(+y\)
coupling expressed in the complex coordinates \(z=\tilde\rho+iy\).
In particular, \eqref{eq:eom-non-hermitian} is \emph{not} complex--linear off the zero--residual leaf:
it has (the Bogoliubov) form that mixes \(z\) and \(\bar z\).

On the zero--residual leaf \(y=0\) the subspace \(y\equiv 0\) is invariant
(\(\dot y=-\widehat V^{\!\top}y\)), hence \(z=\bar z\) and \eqref{eq:eom-non-hermitian} reduces to
\[
\dot z=(\widehat S+\widehat F)z=\widehat V z.
\]
Thus the on--shell dynamics become complex--linear and can be written in a (generally nonunitary)
Schr\"odinger form
\begin{equation} \label{eq:schrodinger-complex-packaging}
\dot z=-\,i\,\widehat H_{\mathbb C}\,z,
\qquad
\widehat H_{\mathbb C}:= i\,\widehat V = i(\widehat S+\widehat F).
\end{equation}
However, \(\widehat H_{\mathbb C}\) is \emph{non--Hermitian} whenever \(\widehat S\neq 0\):
indeed \(\widehat H_{\mathbb C}^\dagger=-\,i\,\widehat V^{\!\top}\neq \widehat H_{\mathbb C}\) in general.
Off the zero--residual leaf (\(y\neq 0\)), the coupling to \(\bar z\) in \eqref{eq:eom-non-hermitian}
precludes any reduction to a unitary complex Schr\"odinger evolution.

This extrinsic complex packaging also makes clear why the usual Hermitian norm need not be conserved.
Differentiating \(\|z\|^2=z^\dagger z\) along solutions of \eqref{eq:eom-non-hermitian} yields
\[
\frac{d}{dt}\|z\|^2
= z^\dagger\!\Bigl(\widehat S+\tfrac{i}{2}I\Bigr)\bar z
+ \bar z^{\!\dagger}\!\Bigl(\widehat S-\tfrac{i}{2}I\Bigr) z
=2\,\Re\!\Bigl(z^\dagger\!\Bigl(\widehat S+\tfrac{i}{2}I\Bigr)\bar z\Bigr),
\]
which is generically nonzero.  In summary, identifying \(\RR^{2n}\simeq\mathbb C^n\) provides an
elegant coordinate description of the \emph{same} lifted real dynamics, but it obscures the
para--K\"ahler geometry: unless the real generator preserves the chosen almost complex structure \(J\),
the resulting representation is not a unitary Schr\"odinger evolution, and even on the zero--residual
leaf the apparent Hamiltonian \(\widehat H_{\mathbb C}\) is non--Hermitian whenever \(\widehat S\neq 0\).
This contrast will be further elaborated in \autoref{sec:4}.

\section{Elliptic quantum geometry and K\"ahler structure}
\label{sec:4}

The preceding sections show that the least--action rationality (LAR) lift admits
a natural \emph{hyperbolic} quantum--geometric reading.  Starting from the
information--geometric (Hessian) data, the cotangent lift
$
N:=T^*\widetilde{\mathcal M}
$
comes with its canonical symplectic form \(\Omega\) and a canonical
para--K\"ahler (para--Hermitian) way of \emph{organising} that symplectic space:
a product structure \(K\) with \(K^2=\id\) splits directions into two
complementary sectors, and a neutral metric \(G\) relates the splitting to the
symplectic form via
\[
\Omega(u,v)=G(Ku,v).
\]
This gives an intrinsic Witt decomposition
$
TN=E_+\oplus E_-,
$
so the hyperbolic lift does not merely say that the dynamics are Hamiltonian on
\((N,\Omega)\); it also supplies a canonical ``two--sector'' geometry in which
the split--complex Schr\"odinger picture becomes a convenient shorthand (most
transparently on the zero--residual sheet).

By contrast, much of the quantum cognition literature (e.g.\ \cite{BusemeyerBruza2012})
assumes an explicitly \emph{elliptic} quantum geometry: amplitudes live in a
complex Hilbert space and evolution is unitary.  From the present perspective,
the key change is not the symplectic data (since \((N,\Omega)\) is the same) but
the choice of \emph{coherence structure} used to interpret the lifted geometry (symplectic space).  In the
hyperbolic theory, coherence is organised by the canonical para--complex split
(\(K^2=+\id\)), leading at the linear level to the noncompact symmetry group
\[
\mathrm{Sp}(2n,\RR)\cap\mathrm{O}(n,n)\cong \mathrm{GL}(n,\RR).
\]
An elliptic model instead requires a K\"ahler organisation \((g,J,\Omega)\) with
\(J^2=-\id\), \(g\) positive definite, and
\[
\Omega(u,v)=g(Ju,v),
\]
whose linear symmetry group is the compact unitary group
\[
\mathrm{Sp}(2n,\RR)\cap\mathrm{O}(2n)\cong \mathrm{U}(n).
\]
The intuitive point is simple: hyperbolic structure is \emph{delivered by the
	lift}, whereas elliptic structure amounts to an extra modelling choice: a way
of selecting a complex notion of orthogonality/coherence on the same symplectic
space.

The purpose of this section is to make that additional choice explicit, and to
show how standard complex (unitary) quantum dynamics can be recovered from the
classical Hamiltonian framework.  Rather than postulating a complex structure
\emph{ad hoc}, we proceed variationally.  Starting from the Hamilton action on
\((N,\Omega)\), we impose an \emph{eigenbundle admissibility restriction}:
admissible trajectories are required to evolve within a distinguished
half--dimensional Lagrangian subbundle of the (complexified) tangent bundle,
while admissible variations are confined to the complementary Lagrangian
directions.  As shown below, this replaces the full Hamilton equation by a
projected evolution and is equivalent to choosing a K\"ahler polarisation.  In
this sense, polarisation is not an independent geometric postulate, but the
symplectic encoding of an admissibility constraint imposed at the level of the
variational principle.

\subsection{Elliptic quantum geometry}
\label{subsec:elliptic}

We keep the same lifted manifold \(N:=T^*\widetilde{\mathcal M}\simeq\RR^{2n}\) and its canonical
symplectic form \(\Omega\).  Since \(\Omega\) is exact, there exists a globally defined one--form
\(\theta\) on \(N\) with \(\Omega=-\,d\theta\).  Any smooth function \(H:N\to\RR\) then determines
a Hamiltonian vector field \(X_H\in\mathfrak X(N)\) by
\begin{equation}\label{eq:XH-def-elliptic}
	\iota_{X_H}\Omega=dH,
\end{equation}
equivalently, the integral curves \(z(t)\) satisfy \(\dot z(t)=X_H(z(t))\).  These same curves are
also characterised variationally as stationary points (with fixed endpoints) of the Hamilton action
\begin{equation}\label{eq:Ham-action-elliptic}
	\mathcal S[z]:=\int_{t_0}^{t_1}\bigl(\theta(\dot z(t)) - H(z(t))\bigr)\,dt.
\end{equation}
The purpose of the present discussion is to explain how an \emph{elliptic} (Hilbert--space) quantum
dynamics can be obtained from this \emph{same} real Hamiltonian \emph{structure} \((N,\Omega,H)\) by adding
an admissibility structure that selects a complex Lagrangian sector.  This additional structure is
not canonically fixed by \(\Omega\) (nor by the information--geometric lift) and therefore
constitutes genuine modelling input.

We consider the global Darboux chart $(\tilde\rho,y)\in\RR^n\oplus\RR^n$ on $N$, so that
$\Omega=d\tilde\rho^i\wedge dy_i$.  In these coordinates we consider the (formal) complexification
$N_\CC\simeq\CC^{2n}$ and its complexified tangent bundle $TN_\CC:=TN\otimes_\RR\CC$, and we extend
$\Omega$ $\CC$--bilinearly to $\Omega_\CC:=\Omega\otimes 1$.  Fix a complex symmetric matrix
$M\in\Sym(n,\CC)$ with $\Im M$ nondegenerate and define constant complex subbundles of $TN_\CC$
(by identifying each fibre with $\CC^{2n}$ in Darboux coordinates) as graphs
\begin{equation}\label{eq:LM-def}
	L_M^{(+)}:=\{(u,Mu):u\in\CC^n\},
	\qquad
	L_M^{(-)}:=\{(u,\bar M u):u\in\CC^n\}.
\end{equation}
The following elementary statement shows that this choice provides a complex bi--Lagrangian
splitting of $TN_\CC$.

\begin{proposition}[Complex bi--Lagrangian splitting determined by $M$]\label{prop:LM-splitting}
	The subbundles $L_M^{(\pm)}\subset TN_\CC$ are complex Lagrangian for $\Omega_\CC$.  Moreover,
	$\Im M$ nondegenerate implies transversality and hence a direct sum decomposition
	\begin{equation}\label{eq:LM-splitting}
		TN_\CC=L_M^{(+)}\oplus L_M^{(-)}.
	\end{equation}
	In particular, the projector $\Pi_M^{(+)}:TN_\CC\to L_M^{(+)}$ onto $L_M^{(+)}$ along
	$L_M^{(-)}$ is well--defined and fibrewise complex linear.
\end{proposition}

\begin{proof}
	For $u,v\in\CC^n$ one computes, using the standard complex bilinear pairing
	$\langle\cdot,\cdot\rangle$ on $\CC^n$ (no conjugation) and the induced formula
	$\Omega_\CC((a,b),(c,d))=\langle a,d\rangle-\langle c,b\rangle$,
	\[
	\Omega_\CC\bigl((u,Mu),(v,Mv)\bigr)
	=\langle u,Mv\rangle-\langle v,Mu\rangle
	=\langle u,(M-M^\top)v\rangle
	=0,
	\]
	since $M^\top=M$.  Thus $L_M^{(+)}$ is isotropic of complex rank $n$, hence Lagrangian; the
	same argument applies to $L_M^{(-)}$.  If $(u,Mu)=(v,\bar M v)$ then $u=v$ and $(M-\bar M)u=0$.
	Since $M-\bar M=2i\,\Im M$ and $\Im M$ is nondegenerate, it follows that $u=0$, hence
	$L_M^{(+)}\cap L_M^{(-)}=\{0\}$.  As both are rank $n$ complex subbundles of rank $2n$,
	this implies \eqref{eq:LM-splitting}.  The projector exists and is smooth because the splitting
	is constant in the chosen Darboux chart.
\end{proof}

For later use, it is convenient to record an explicit formula for $\Pi_M^{(+)}$ in Darboux
coordinates.  Given $(a,b)\in\CC^n\oplus\CC^n$, the unique decomposition
$(a,b)=(u,Mu)+(v,\bar M v)$ is determined by $u+v=a$ and $Mu+\bar M v=b$, equivalently
$(M-\bar M)u=b-\bar M a$.  Since $M-\bar M=2i\,\Im M$ is invertible, one finds
\begin{equation}\label{eq:PiM-explicit}
	u=(2i\,\Im M)^{-1}(b-\bar M a),
	\qquad
	\Pi_M^{(+)}(a,b)=(u,Mu).
\end{equation}

We now impose a complex eigenbundle admissibility restriction directly at the level of the
variational principle on $N_\CC$.  A $C^1$ curve $z:[t_0,t_1]\to N_\CC$ is called admissible if
\begin{equation}\label{eq:admissible-kinematics}
	\dot z(t)\in L_M^{(+)}\quad\text{for all }t\in[t_0,t_1],
\end{equation}
and an admissible variation along $z(\cdot)$ is a $C^1$ vector field $\delta z(t)\in T_{z(t)}N_\CC$
with fixed endpoints and values constrained by
\begin{equation}\label{eq:admissible-variations}
	\delta z(t)\in L_M^{(-)}\quad\text{for all }t\in[t_0,t_1],
	\qquad
	\delta z(t_0)=\delta z(t_1)=0.
\end{equation}
Stationarity of the Hamilton action under the restricted variations replaces the full Hamilton
equation by its projection onto the admissible Lagrangian sector.

\begin{proposition}[Projected Hamilton equation from complex eigenbundle admissibility]\label{prop:projected-Ham-M}
	Let $H:N\to\RR$ be real--analytic (in particular polynomial in Darboux coordinates), and denote
	by the same symbol its holomorphic extension to $N_\CC$.  Let $X_H$ be the complex Hamiltonian
	vector field on $N_\CC$ defined by $\iota_{X_H}\Omega_\CC=dH$.  A $C^1$ admissible curve $z(\cdot)$
	satisfying \eqref{eq:admissible-kinematics} is stationary for \eqref{eq:Ham-action-elliptic} under
	all admissible variations \eqref{eq:admissible-variations} if and only if it satisfies
	\begin{equation}\label{eq:projected-Ham-M}
		\dot z(t)=\Pi_M^{(+)}\,X_H\bigl(z(t)\bigr)\in L_M^{(+)}\big|_{z(t)},
		\qquad
		\textnormal{for all }t\in[t_0,t_1].
	\end{equation}
\end{proposition}

\begin{proof}
	We extend $\theta$ and $\Omega$ $\CC$--bilinearly to $N_\CC$ and consider variations $z_\varepsilon$
	with real parameter $\varepsilon$, writing $\delta z=\frac{d}{d\varepsilon}\big|_{\varepsilon=0}z_\varepsilon$.
	Stationarity means $\frac{d}{d\varepsilon}\big|_{\varepsilon=0}\mathcal S[z_\varepsilon]=0$ in $\CC$.
	The first variation of \eqref{eq:Ham-action-elliptic} on $N_\CC$ is the standard identity
	\[
	\delta\mathcal S[z]
	=
	\Bigl[\theta(\delta z)\Bigr]_{t_0}^{t_1}
	+
	\int_{t_0}^{t_1}\Omega_\CC\bigl(\delta z,\dot z-X_H(z)\bigr)\,dt,
	\]
	which follows from $\Omega_\CC=-\,d\theta$ and $\iota_{X_H}\Omega_\CC=dH$.
	The boundary term
	vanishes by \eqref{eq:admissible-variations}, hence stationarity for all admissible $\delta z$
	is equivalent to
	\[
	\Omega_\CC\bigl(\delta z(t),\dot z(t)-X_H(z(t))\bigr)=0
	\qquad
	\forall\,t\in[t_0,t_1],\ \forall\,\delta z(t)\in L_M^{(-)}.
	\]
	Decompose $e(t):=\dot z(t)-X_H(z(t))$ according to \eqref{eq:LM-splitting} as
	$e(t)=e_+(t)+e_-(t)$ with $e_\pm(t)\in L_M^{(\pm)}$.  Since $L_M^{(-)}$ is Lagrangian,
	$\Omega_\CC(\delta z,e_-)=0$ for all $\delta z\in L_M^{(-)}$, and the condition reduces to
	$\Omega_\CC(\delta z,e_+)=0$ for all $\delta z\in L_M^{(-)}$.  Transversality implies that
	$\Omega_\CC$ pairs $L_M^{(-)}$ nondegenerately with $L_M^{(+)}$, hence $e_+(t)=0$ for all $t$,
	i.e.\ $\Pi_M^{(+)}e(t)=0$.  Applying $\Pi_M^{(+)}$ and using $\dot z(t)\in L_M^{(+)}$ yields
	$\dot z(t)=\Pi_M^{(+)}X_H(z(t))$, which is \eqref{eq:projected-Ham-M}.
\end{proof}

In the Darboux chart \eqref{eq:LM-def}, the kinematic constraint \eqref{eq:admissible-kinematics}
is the graph relation
\begin{equation}\label{eq:graph-kinematics}
	\dot y(t)=M\,\dot{\tilde\rho}(t).
\end{equation}
Thus Proposition~\ref{prop:projected-Ham-M} provides the general reduction principle: the
variational admissibility restriction converts Hamilton's equation into a projected evolution
tangent to the chosen complex Lagrangian sector.  The remainder of the elliptic construction is
then an explicit analysis of this projected dynamics in coordinates adapted to the graph
constraint.

When $\Im M$ is definite, one may (after a constant real symplectic change of Darboux
coordinates, i.e.\ an $\mathrm{Sp}(2n,\RR)$--change of linear coordinates preserving $\Omega$) normalise
$\Im M=-I$ and write $M=R-iI$ with $R^\top=R$.  In this normalised class, define the
complex linear combinations
\begin{equation}\label{eq:psi-phi-def}
	\psi:=(I-iR)\tilde\rho+i\,y,
	\qquad
	\phi:=(I+iR)\tilde\rho-i\,y.
\end{equation}
Differentiating and substituting \eqref{eq:graph-kinematics} gives
\begin{equation}\label{eq:phi-const}
	\dot\phi(t)=0,
\end{equation}
so admissible kinematics foliate the constrained dynamics by affine leaves $\{\phi=\phi_0\}$,
and on each leaf the projected equation \eqref{eq:projected-Ham-M} closes as an evolution
equation for the holomorphic coordinate $\psi$.

In the symplectically normalised coordinates \eqref{eq:psi-phi-def}, the admissibility restriction can be stated
equivalently as follows: \(\dot z\in L_M^{(+)}\) is equivalent to the graph relation \eqref{eq:graph-kinematics}, and
in the \((\psi,\phi)\) variables this is precisely the conservation law \(\dot\phi=0\).
 Thus admissible
kinematics foliate the constrained dynamics by affine leaves \(\{\phi=\phi_0\}\), and on each leaf the projected
equation \eqref{eq:projected-Ham-M} closes as a first--order evolution equation for the holomorphic coordinate
\(\psi\).

To connect this general reduction principle with standard complex quantum dynamics, we now specialise to a
distinguished polarisation.  For the lifted least--action Hamiltonian \eqref{eq:hamiltonian-fn}, the choice
$
M_\ast=(1-i)I
$
(i.e.\ \(R=I\) in \eqref{eq:psi-phi-def}) yields, on the invariant holomorphic leaf \(\{\phi=0\}\), an elliptic
Schr\"odinger equation:
\begin{proposition}[Schr\"odinger dynamics from the distinguished polarisation]
	\label{prop:Schrodinger-from-Mstar}
	Let
	\(
	H(\tilde\rho,y)=\tfrac12\|y\|^2+\langle y,\widehat V\tilde\rho\rangle
	\)
	be the lifted least--action Hamiltonian with
	\(\widehat V=\widehat S+\widehat F\),
	\(\widehat S^\top=\widehat S\), \(\widehat F^\top=-\widehat F\).
	Fix the complex polarisation determined by
	\[
	M_\ast=(1-i)I,
	\qquad \Im M_\ast=-I,
	\]
	and use the associated holomorphic/antiholomorphic coordinates (cf.~\eqref{eq:psi-phi-def} with \(R=I\))
	\[
	\psi:=(1-i)\tilde\rho+i\,y,
	\qquad
	\phi:=(1+i)\tilde\rho-i\,y.
	\]
	Then each affine leaf \(\{\phi=\phi_0\}\) is invariant under the projected Hamilton dynamics
	\eqref{eq:CLAR-projected-Hamilton}; in particular \(\{\phi=0\}\) is invariant.  On \(\{\phi=0\}\) one has
	\begin{equation}\label{eq:Schrodinger-Mstar-gauged}
		i\,\dot\psi=(I+\widehat S+i\widehat F)\psi.
	\end{equation}
	Equivalently, after the scalar rephasing \(\Psi(t):=e^{it}\psi(t)\), the evolution is the standard
	Schr\"odinger equation
	\begin{equation} \label{eq:schrödinger}
	i\,\dot\Psi=(\widehat S+i\widehat F)\Psi.
	\end{equation}
\end{proposition}

\begin{proof}
	For \(H(\tilde\rho,y)=\tfrac12\|y\|^2+\langle y,\widehat V\tilde\rho\rangle\) the Hamiltonian vector field is
	\[
	X_H(\tilde\rho,y)=\bigl(a,b\bigr)
	=\bigl(y+\widehat V\tilde\rho,\,-\widehat V^{\!\top}y\bigr).
	\]
	The projected dynamics \eqref{eq:CLAR-projected-Hamilton} is
	\(\dot z=\Pi^{(+)}_{M_\ast}X_H(z)\), hence \(\dot z\in L^{(+)}_{M_\ast}\) and therefore
	\(\dot y=M_\ast\dot{\tilde\rho}\).  Differentiating the definitions of \(\phi\) and \(\psi\) gives
	\[
	\dot\phi=(1+i)\dot{\tilde\rho}-i\dot y=(1+i-iM_\ast)\dot{\tilde\rho}=0,
	\qquad
	\dot\psi=(1-i)\dot{\tilde\rho}+i\dot y=(1-i+iM_\ast)\dot{\tilde\rho}=2\dot{\tilde\rho}.
	\]
	Hence each leaf \(\{\phi=\phi_0\}\) is invariant and \(\dot{\tilde\rho}=\tfrac12\dot\psi\).
	
	Using the explicit projector formula \eqref{eq:PiM-explicit} with \(\bar M_\ast=(1+i)I\) and \(\Im M_\ast=-I\),
	the projected velocity satisfies
	\[
	\dot{\tilde\rho}
	=(2i\,\Im M_\ast)^{-1}\bigl(b-\bar M_\ast a\bigr)
	=\frac{1}{2i}\bigl(\bar M_\ast a-b\bigr).
	\]
	On the holomorphic leaf \(\{\phi=0\}\) one has \(y=M_\ast\tilde\rho\) and \(\tilde\rho=\psi/2\), hence
	\(a=\tfrac12(M_\ast+\widehat V)\psi\) and \(b=-\tfrac12\widehat V^{\!\top}M_\ast\psi\).  Substitution yields
	\[
	\dot\psi=2\dot{\tilde\rho}
	=\frac{1}{i}\Bigl(\bar M_\ast a-b\Bigr)
	=\frac{1}{i}\Bigl(I+\widehat S+i\widehat F\Bigr)\psi,
	\]
	i.e.\ \eqref{eq:Schrodinger-Mstar-gauged}.  The final Schr\"odinger form follows by the rephasing
	\(\Psi(t)=e^{it}\psi(t)\).
\end{proof}
To summarise, in the present \emph{elliptic} construction one supplies, as
additional modelling input, an \(\Omega\)--compatible complex Lagrangian sector
\(L_M^{(+)}\subset TN_\CC\).  In the adapted coordinates \((\psi,\phi)\) this
admissibility requirement is equivalent to the conservation law \(\dot\phi=0\).
Accordingly, the projected dynamics is foliated by invariant affine leaves
\(\{\phi=\phi_0\}\).  For the distinguished choice \(M_\ast=(1-i)I\), the leaf
\(\{\phi=0\}\) (vanishing anti--holomorphic coordinate) carries a Schr\"odinger
evolution with generator \(\widehat S+i\widehat F\), up to an additive scalar
multiple of the identity (equivalently, a global phase / projective gauge).

It is useful to contrast this with the \emph{hyperbolic} (para--K\"ahler) lift of
\autoref{sec:3}.  There the canonical Witt splitting is fixed by the
Fisher--Rao/Hessian data, and the deviation from para--unitarity is carried
entirely by the nilpotent shear.  On the {zero--residual leaf}
$
\Sigma_0=\{(\tilde\rho,y):y=0\},
$
the shear is inactive, the dynamics closes as a para--unitary block flow, and
one obtains a para--holomorphic (split--complex, \(\mathbb D\)) Schr\"odinger
representation (the ``para--Schr\"odinger'' equation (\ref{eq:para-schrodinger}).  Thus both sectors admit
a natural ``zero leaf'': \(\{\phi=0\}\) in the elliptic case and \(\Sigma_0\) in
the hyperbolic case.  The distinction is that \(\Sigma_0\) is canonical
from the lift itself, whereas \(\{\phi=0\}\) depends on the chosen complex
sector \(L_M^{(+)}\).

Consequently, elliptic coherence does \emph{not} follow from least--action
rationality alone: it requires an additional admissibility restriction selecting
a complex Lagrangian sector.  This motivates a modified least--action postulate
\emph{within} the selected holomorphic sector, which we formulate next as
coherent least--action rationality (CLAR).

\subsection{Coherent least--action rationality}
\label{subsec:CLAR}

The behavioural assumptions in \autoref{sec:2}, in particular the assumption on LAR in \autoref{subsec:least-action-rationality}, determine a canonical para--K\"ahler lift on
\(N:=T^*\widetilde{\mathcal M}\) and yield the lifted hyperbolic (Krein--space) model as derived in \autoref{sec:3}.  By contrast, an \emph{elliptic} (Hilbert--space) quantum model requires
additional structure not fixed by the Fisher--Rao data: as we saw above in \autoref{subsec:elliptic}, one must specify an \(\Omega\)--compatible complex
sector, equivalently a complex Lagrangian polarisation, and then restrict admissible evolutions to that
sector. In behavioural terms, elliptic coherence is therefore not an implication of Assumptions~A1--A3,
but an additional admissibility requirement.

We formulate this additional requirement as a modified least--action postulate on the same underlying
real symplectic phase space.  Let \(\theta\) be the tautological one--form on \(N\) with \(\Omega=-\,d\theta\),
let \(H:N\to\RR\) be the lifted Hamiltonian, and extend \((N,\Omega,\theta,H)\) to the
complexification \(N_\CC\).  Fix a complex bi--Lagrangian splitting of the complexified tangent bundle,
\[
TN_\CC=L^{(+)}\oplus L^{(-)},
\qquad
\Omega_\CC|_{L^{(\pm)}}\equiv 0,
\qquad
L^{(-)}=\overline{L^{(+)}}.
\]
In the constant Darboux setting used in \autoref{subsec:elliptic}, such a splitting is equivalently specified
by a matrix \(M\in\Sym(n,\CC)\) with \(\Im M\) nondegenerate via
\(L^{(+)}=L_M^{(+)}\), \(L^{(-)}=L_M^{(-)}\) as in \eqref{eq:LM-def}, with associated projector
\(\Pi^{(+)}=\Pi_M^{(+)}\) onto \(L^{(+)}\) along \(L^{(-)}\).

Consider the Hamilton action on \(C^1\) curves \(z:[t_0,t_1]\to N_\CC\),
\[
\mathcal S[z]:=\int_{t_0}^{t_1}\bigl(\theta(\dot z(t)) - H(z(t))\bigr)\,dt.
\]
A curve \(z(\cdot)\) is called \emph{coherent admissible} (for the chosen polarisation) if
\begin{equation}\label{eq:CLAR-admissible-curves}
	\dot z(t)\in L^{(+)}_{z(t)}
	\qquad\textnormal{for all }t\in[t_0,t_1],
\end{equation}
and an \emph{admissible variation} along \(z(\cdot)\) is a \(C^1\) field \(\delta z(t)\in T_{z(t)}N_\CC\) with
fixed endpoints and values constrained by
\begin{equation}\label{eq:CLAR-admissible-variations}
	\delta z(t)\in L^{(-)}_{z(t)},
	\qquad
	\delta z(t_0)=\delta z(t_1)=0.
\end{equation}

\medskip
\noindent
\textbf{Assumption 3'} [Coherent least--action rationality (CLAR)]\label{assump:CLAR}
\textit{
	Given a choice of complex polarisation \(TN_\CC=L^{(+)}\oplus L^{(-)}\), the decision maker realises
	coherent admissible episodes \(z(\cdot)\) that are stationary for the Hamilton action \(\mathcal S[z]\)
	under all admissible variations \eqref{eq:CLAR-admissible-variations}.
}

\medskip

By Proposition~\ref{prop:projected-Ham-M}, Assumption 3' is equivalent to the
\emph{projected Hamilton equation}
\begin{equation}\label{eq:CLAR-projected-Hamilton}
	\dot z(t)=\Pi^{(+)}X_H\bigl(z(t)\bigr)\in L^{(+)}_{z(t)},
	\qquad
	t\in[t_0,t_1],
\end{equation}
where \(X_H\) is the (complex) Hamiltonian vector field defined by \(\iota_{X_H}\Omega_\CC=dH\).
Thus CLAR replaces the full Hamilton evolution by its polarisation component: the complementary
\(L^{(-)}\) component of \(X_H\) is not followed, but excluded by admissibility.

When \(\Im M\) is definite (after normalization \(\Im M<0\)), the polarisation induces a positive definite
sesquilinear form on \(L^{(+)}\) via
\(\langle u,v\rangle_M:=-\,i\,\Omega_\CC(u,\bar v)\),
so that \eqref{eq:CLAR-projected-Hamilton} becomes a unitary evolution on the corresponding Hilbert
space.  For the lifted least--action Hamiltonian \eqref{eq:ham-system}, the distinguished choice
\(M_\ast=(1-i)I\) yields, on the invariant holomorphic leaf \(\{\phi=0\}\), the standard Schr\"odinger
dynamics with Hermitian generator \(\widehat S+i\widehat F\) (Proposition~\ref{prop:Schrodinger-from-Mstar}).
In particular, coherence in the elliptic model is an explicit behavioural restriction: it is imposed by
selecting a complex polarisation and applying the least--action principle \emph{within} its holomorphic
sector, rather than emerging from the hyperbolic LAR lift alone.

\section{Bounded Rationality}
\label{sec:BR}

This section reviews the least--action rationality (LAR) model through the lens of \emph{bounded rationality}.  Our guiding
theme is the separation between the latent information--geometric dynamics, formulated on amplitudes (and their lift), and the
epistemic readout on the simplex that produces observable choice frequencies.  We will primarily consider the hyperbolic
(split--complex/para--K\"ahler) regime developed above, where the lifted phase space
\[
N:=T^*\widetilde{\mathcal M}\subset T^*\RR^n
\]
carries its canonical symplectic form \(\Omega\) together with the para--K\"ahler structure induced by the Hessian data, and the
Hamiltonian LAR dynamics on \(N\) restricts on the zero--residual leaf to the para--unitary (split--complex Schr\"odinger--type)
sector.  A final subsection (\autoref{subsec:BR-elliptic-quantum}) returns briefly to the elliptic (complex, unitary) coherent
sector to contrast which bounded--rationality mechanisms persist under CLAR and which are specific to the hyperbolic lift.

As in other quantum--like approaches, two levels of state description coexist.  The
\emph{latent state} is an amplitude (or lifted amplitude--momentum pair) evolving on the
information--geometric lift, while the \emph{epistemic state} is a probability distribution on
the simplex $\mathcal M=\Delta^{n-1}_{>0}$ obtained by geometric projection (normalisation and squaring in a
chosen readout basis).  The latent variables need not be endowed with an ontological
interpretation; what matters for our purposes is that they support a canonical differential
geometry and a preferred LAR dynamics, whereas the epistemic variables are what is observable
as choice frequencies in a fixed context.  Bounded rationality is then understood as the
mismatch between these two descriptions: even when the latent layer obeys a clean variational
principle, the induced epistemic behaviour can be path dependent, history dependent, and
context dependent because of structural constraints in the latent generator and because of the
geometric reductions required to pass to observables.

We separate these mechanisms according to \emph{where} they arise: intrinsically on the epistemic simplex \(\mathcal M\),
intrinsically in the latent on--shell flow on the zero--residual leaf, intrinsically in the full lifted (off--shell)
phase--space dynamics, or from the geometric readout map that reduces latent amplitudes to observable probabilities.
Accordingly, we proceed as follows.  \autoref{subsec:non-integrability} studies epistemic non--integrability of the preference
form \(\beta\) on \(\mathcal M\).  \autoref{subsec:entropic-clock} studies the latent on--shell split
\(\widehat V=\widehat S+\widehat F\) on the zero--residual leaf and introduces a monotone entropic clock on rays.
\autoref{subsec:neutral-cone} studies off--shell deviation using the neutral form \(Q\) and the associated index \(\Lambda(t)\).
\autoref{subsec:geometry-readout} studies context dependence and hyperbolic interference arising from the readout family \(\pi_B\).
We conclude in \autoref{subsec:BR-elliptic-quantum} by returning briefly to the elliptic coherent sector (CLAR), to contrast
which bounded--rationality mechanisms persist under unitary Hilbert--space evolution and which are specific to the hyperbolic
lift.

\subsection{Non-integrability}
\label{subsec:non-integrability}

Assume $n\ge3$, so $\mathcal M=\Delta^{n-1}_{>0}$ has $\dim\mathcal M\ge2$ and admits nonconstant smooth loops and embedded disks.
Recall from \autoref{subsec:PreferenceOperators-full} that the induced epistemic preference field on lotteries is the
$1$--form $\beta=\mathcal T(\widehat V)\in\Omega^1(\mathcal M)$ obtained by restricting the spherical form
$\bar\alpha^{\mathbb S}$ to $\mathbb S^{n-1}_{+}$ and pulling it back along the square--root chart
$\iota:\mathcal M\to\mathbb S^{n-1}_{+}$, $\iota(q)=\sqrt q$.
All simplex--level differential identities involving $\beta$ (and hence $d\beta$) are understood on $\mathcal M$ and,
if a trajectory touches the boundary of $\Delta^{n-1}$, are to be read piecewise on each time interval of constant support
as in \autoref{rem:stratified-reading}.
By Proposition~\ref{prop:beta-decomp} this field splits canonically as
\[
\beta=\mathcal U+\mathcal R,
\qquad
\mathcal U=dU,
\qquad
\mathcal R=\beta-\mathcal U,
\qquad
d\beta=d\mathcal R.
\]
The conservative part $\mathcal U$ admits a global potential, while $\mathcal R$ captures the obstruction to global
integrability.  We package this obstruction as follows.

\begin{definition}[Co--utility bounded rationality]\label{def:co-utility}
	Epistemic preferences exhibit \emph{co--utility bounded rationality} if $\beta$ is not globally representable by a utility
	potential on $\mathcal M$, i.e.\ if there is no $U:\mathcal M\to\RR$ with $\beta=dU$.
	Since $\mathcal M$ is contractible, this is equivalent to $d\beta\not\equiv0$, equivalently to the existence of a smooth closed
	loop $\gamma\subset\mathcal M$ with $\oint_\gamma \beta\neq 0$ (and by Stokes' theorem
	$\oint_{\partial D}\beta=\int_D d\beta$ for any embedded $2$--disk $D$ spanning $\gamma$ and contained in $\mathcal M$; under
	boundary contacts this identity is read on each fixed--support stratum as in \autoref{rem:stratified-reading}).
\end{definition}

Following the geometric viewpoint of \cite{Russell1991}, we therefore interpret the $2$--form
\[
\mathcal K:=d\beta
\]
as a \emph{preference curvature}: it measures the failure of path independence, and its flux controls the accumulated loop
holonomy (regret) generated by sequential changes.  In our operator class the conservative part is precisely the utility channel,
so the obstruction is carried entirely by the co--utility channel: $d\beta=d\mathcal R$ by Proposition~\ref{prop:beta-decomp}.
Thus expected--utility behaviour corresponds to $\mathcal K\equiv0$, whereas nonzero curvature predicts path dependence and order
effects in sequential choice.  This packaging is consistent with classical regret/disappointment motivations
\cite{LoomesSugden1982,Bell1982}. The nonconservative component arises intrinsically as the
co--utility component induced by the same latent operator that also generates the utility potential.

\begin{remark}
	The discussion above can be viewed as an instantiation of a ``tower of irrationality''. As suggested by \cite{Russell1991}:
	one starts from a preference $1$--form $\beta$ and measures successive departures from classical potential representation by
	iterating the exterior derivative,
	\[
	\beta,\quad d\beta,\quad d(d\beta),\ \dots
	\]
	In the present setting the first step already exhausts the meaningful obstruction for epistemic preferences.  First, because
	$\mathcal M=\Delta^{n-1}_{>0}$ is contractible, the only de Rham obstruction to a global utility potential is the curvature
	$d\beta$; once $d\beta$ vanishes, $\beta$ is closed and hence exact.  Second, the exterior derivative tower always terminates
	immediately after the first obstruction for $1$--forms, since $d^2=0$ implies $d(d\beta)\equiv 0$.  What is model--dependent is
	therefore not the formal truncation, but the interpretation: in our construction the nonintegrable content of $\beta$ is
	canonically identified with the co--utility channel induced by the skew--symmetric part $\widehat F$ of the latent preference
	operator, and the preference curvature $d\beta$ measures the resulting holonomy (regret) around decision loops.
\end{remark}

\subsection{The entropic welfare clock on the zero--residual leaf}
\label{subsec:entropic-clock}

Recall that in our operator representation \(\widehat V=\widehat S+\widehat F\) the symmetric channel \(\widehat S\) encodes
evaluative norm change, while the skew channel \(\widehat F\) generates circulatory (regret--like) drift.  The purpose of the
present subsection is to make this separation explicit in polar variables and to introduce a canonical scalar clock on rays,
anchored by a spectral gauge choice that will later serve as a reference when we leave the zero--residual leaf.  The clock is
called \emph{entropic} in the sense that it is the log--partition normaliser in the quadratic readout (hence paired with
Shannon entropy by convex duality), and in the sense that its on--shell \emph{production rate} is generated by the evaluative
channel \(\widehat S\), not by the circulatory channel \(\widehat F\).

Recall that on the zero--residual leaf we set \(y\equiv 0\), so the lifted amplitude satisfies the linear on--shell dynamics
\eqref{eq:ham-system}, i.e.\ \(\dot{\tilde\rho}=\widehat V\,\tilde\rho\) with \(\widehat V=\widehat S+\widehat F\).
We write the polar decomposition
\[
r(t):=\|\tilde\rho(t)\|>0,
\qquad
\rho(t):=\frac{\tilde\rho(t)}{\|\tilde\rho(t)\|}\in\mathbb S^{n-1},
\qquad
\tilde\rho(t)=r(t)\rho(t).
\]
Differentiating \(\|\rho(t)\|^2\equiv 1\) yields \(\ip{\rho(t)}{\dot\rho(t)}=0\).  Substituting
\(\dot{\tilde\rho}=\dot r\,\rho+r\,\dot\rho\) into the on--shell evolution \(\dot{\tilde\rho}=\widehat V\tilde\rho\)
and separating radial and tangential components gives the polar identities
\begin{equation}\label{eq:r-dot-onshell}
	\dot r(t)=\ip{\rho(t)}{\widehat S\,\rho(t)}\,r(t),
\end{equation}
and
\begin{equation}\label{eq:rho-dot-onshell}
	\dot\rho(t)
	=\Bigl(\widehat F+\widehat S-\ip{\rho(t)}{\widehat S \rho(t)}\,I\Bigr)\rho(t),
	\qquad
	\ip{\rho(t)}{\dot\rho(t)}=0.
\end{equation}

The channel separation is already visible at the level of norm change: by skew--symmetry,
\(\ip{\rho}{\widehat F \rho}=0\) for all \(\rho\), so \(\widehat F\) contributes no \emph{direct} radial component, whereas
\(\widehat S\) completely determines the radial rate \(\dot r/r\) via \eqref{eq:r-dot-onshell}.  At the same time,
\(\widehat S\) also affects the tangential drift through \((\widehat S-\ip{\rho}{\widehat S\rho}I)\rho\), and the
instantaneous direction \(\rho(t)\) influences the radial rate through the Rayleigh quotient
\(\ip{\rho(t)}{\widehat S\rho(t)}\).  Circulation may therefore reshape the path \(\rho(t)\), but it cannot generate clock
production when \(\widehat S=0\).

The epistemic state is obtained by squaring coordinates,
\[
q_i(t):=\rho_i(t)^2,
\qquad
q(t)\in\Delta^{n-1},
\]
so simplex observables depend only on the ray \(\RR_{>0}\tilde\rho(t)\), not on the scale \(r(t)\).  A convenient way to
encode this ray--invariance is the identity--shift invariance of the simplex projection: for any \(c\in\RR\),
\[
\widehat S\mapsto \widehat S+cI
\quad\Longrightarrow\quad
\tilde\rho(t)\mapsto e^{ct}\tilde\rho(t)
\quad\Longrightarrow\quad
q_i(t)=\frac{\tilde\rho_i(t)^2}{\|\tilde\rho(t)\|^2}\ \text{ unchanged},
\]
and the direction equation \eqref{eq:rho-dot-onshell} is unchanged as well, since the shift cancels in
\(\widehat S-\ip{\rho}{\widehat S \rho}I\).  We therefore regard \(\widehat S\mapsto \widehat S+cI\) as a projective gauge.

We now introduce the scale variable that will serve as an intrinsic entropic (equivalently, welfare) clock.  Define the
normaliser
\[
\mathcal Z(t):=\|\tilde\rho(t)\|^{2}=\sum_{i=1}^n \tilde\rho_i(t)^2,
\qquad
\sigma(t):=\log\|\tilde\rho(t)\|=\tfrac12\log \mathcal Z(t),
\]
and the nonnegative weights \(w_i(t):=\tilde\rho_i(t)^2\), so that \(\mathcal Z=\sum_i w_i\) and \(q_i=w_i/\sum_j w_j\).
When \(w_i(t)>0\) we may write \(w_i=e^{\varphi_i}\) with \(\varphi_i(t):=\log w_i(t)\), giving
\[
\mathcal Z(t)=\sum_{i=1}^n e^{\varphi_i(t)},
\qquad
q_i(t)=\frac{e^{\varphi_i(t)}}{\sum_{j=1}^n e^{\varphi_j(t)}}.
\]
(These identities extend to the boundary by continuity, with the convention \(0\log 0:=0\).)  Thus \(\log \mathcal Z=\log\sum_i e^{\varphi_i}\)
is the log--sum--exp normaliser (the ``inclusive value'' in discrete choice), and it is canonically paired with Shannon entropy
by the variational identity
\[
\log\Bigl(\sum_{i=1}^n e^{\varphi_i}\Bigr)
=
\sup_{q\in\Delta^{n-1}}
\Bigl\{\sum_{i=1}^n q_i\varphi_i + \mathcal H(q)\Bigr\},
\qquad
\mathcal H(q):=-\sum_{i=1}^n q_i\log q_i,
\]
(with the convention \(0\log 0:=0\)).  Since \(q(t)\) is precisely the normalised weight vector, the supremum is attained
along the episode, yielding the free--energy identity
\begin{equation}\label{eq:free-energy-identity}
	\log\mathcal Z(t)=\sum_{i=1}^n q_i(t)\,\varphi_i(t)+\mathcal H\!\bigl(q(t)\bigr).
\end{equation}
This is the sense in which \(\sigma=\tfrac12\log\mathcal Z\) is simultaneously ``welfare'' (a log--partition potential) and
``entropic'' (paired with \(\mathcal H\) by convex duality), independently of whether the tangential drift contains
circulation (\(\widehat F\neq 0\)).

For dynamical purposes, it is convenient to anchor \(\sigma\) in a canonical way that respects the projective gauge above.
Let
\[
\lambda_{\min}:=\lambda_{\min}(\widehat S),
\qquad
\widehat S_{+}:=\widehat S-\lambda_{\min}I\succeq 0,
\]
and define the spectrally normalised \emph{entropic welfare clock}
\begin{equation}\label{eq:sigma-plus-def}
	\sigma_{+}(t):=\log\|\tilde\rho(t)\|-\lambda_{\min}\,t
	\;=\;\tfrac12\log \mathcal Z(t)-\lambda_{\min}\,t.
\end{equation}
This anchoring is projectively invariant: under \(\widehat S\mapsto\widehat S+cI\) one has
\(\log\|\tilde\rho(t)\|\mapsto \log\|\tilde\rho(t)\|+ct\) and \(\lambda_{\min}\mapsto\lambda_{\min}+c\), so \(\sigma_+\) is unchanged.
Along any zero--residual trajectory we have
\begin{equation}\label{eq:sigma-plus-dot}
	\dot\sigma_{+}(t)=\ip{\rho(t)}{\widehat S_{+}\rho(t)}\ge 0,
\end{equation}
so \(\sigma_{+}\) is nondecreasing and strictly increasing unless \(\rho(t)\) remains in the ground eigenspace
\(\ker(\widehat S_{+})\).  In particular, the clock production is driven by the evaluative channel: \(\widehat F\) contributes
no direct radial component, and when \(\widehat S=0\) one has \(\dot\sigma_{+}\equiv 0\).  Moreover, \eqref{eq:sigma-plus-dot}
depends only on the ray \(\rho(t)\), so up to an additive constant the clock is determined by the ray trajectory:
\[
\sigma_{+}(t)-\sigma_{+}(t_0)=\int_{t_0}^t \ip{\rho(s)}{\widehat S_{+}\rho(s)}\,ds.
\]
This is consistent with the auxiliary complex repackaging in \autoref{subsec:non-hermitian-packaging}, in which the packaged
generator \(\widehat H_{\C}\) is defined. If \(\widehat H_{\C}\) is hermitian  with respect to the standard complex inner product induced by the
Euclidean metric implies unitary evolution and hence norm preservation.  In particular, in the purely circulatory case
(\(\widehat S=0\)) the evolution is unitary and preserves the Euclidean norm.

Finally, mapping \eqref{eq:rho-dot-onshell} through \(q_i=\rho_i^2\) yields the induced simplex drift.  At the level of the
operator decomposition, this drift separates into an evaluative contribution driven by \(\widehat S\) and a circulatory
contribution driven by \(\widehat F\).  When \(n\ge3\), Proposition~\ref{prop:beta-decomp} identifies this separation on the
simplex by decomposing the induced preference $1$--form \(\beta=\mathcal T(\widehat V)\) into its potential part \(dU\)
(utility) and its complementary part (co--utility) under the regularity assumptions stated there.

The polar variables also anticipate what changes off--shell.  Along a general lifted episode one has
\(\dot{\tilde\rho}=y+\widehat V\tilde\rho\), so the radial rate acquires precisely the radial component of the residual,
\[
\dot r(t)=\ip{\rho(t)}{y(t)}+r(t)\ip{\rho(t)}{\widehat S\rho(t)}.
\]
Equivalently, at the level of the normaliser one finds
\[
\dot{\mathcal Z}(t)=2\,\ip{\tilde\rho(t)}{y(t)}+2\,\ip{\tilde\rho(t)}{\widehat S\,\tilde\rho(t)},
\qquad
\dot\sigma(t)=\ip{\rho(t)}{\widehat S\rho(t)}+\frac{\ip{\tilde\rho(t)}{y(t)}}{\mathcal Z(t)}.
\]
Thus the unique new scalar contribution is measured by the neutral pairing \(\ip{\tilde\rho(t)}{y(t)}\).  In the next
subsection we make this phase--space observable explicit via the null--cone index
\(\Lambda(t)=2\ip{\tilde\rho(t)}{y(t)}\) (so that \(\dot\sigma=\ip{\rho}{\widehat S\rho}+\Lambda/(2\mathcal Z)\)), and show how
it records action accumulation and enforces a one--way cone--crossing constraint for off--shell episodes.

\subsection{Structural hyperbolic bounded rationality: the null cone, action accumulation, and one--way cone crossing}
\label{subsec:neutral-cone}

The preceding subsections isolated bounded rationality mechanisms visible on the epistemic manifold \(\mathcal M\):
non--integrability is detected by the preference curvature \(d\beta\), and on the zero--residual leaf the latent split
\(\widehat V=\widehat S+\widehat F\) separates evaluative norm change from circulatory drift.  We now focus on a distinct,
specifically hyperbolic mechanism that is only visible in the full lifted dynamics: off--shell episodes with nonzero
residual are organised by an orthogonality relation in phase space, encoded by the canonical neutral form of the
para--K\"ahler lift.

Recall that we consider the Darboux chart \(Z=(\tilde\rho,y)\in T^*\widetilde{\mathcal M}\simeq\RR^{2n}\) of the lifted model, where \(y\)
denotes the residual (momentum) and \(t\mapsto Z(t)=(\tilde\rho(t),y(t))\) is a lifted episode satisfying
\eqref{eq:ham-system}.  By \autoref{cor:ambient-parasplitting-geom} and \eqref{eq:neutral-Q}, the Euclidean Hessian lift
carries the canonical neutral quadratic form
\[
Q(\tilde\rho,y)=2\langle \tilde\rho,y\rangle .
\]
Along an episode \(Z(t)=(\tilde\rho(t),y(t))\) we therefore consider the associated scalar observable
\begin{equation}\label{eq:Lambda-def}
	\Lambda(t):=(Q\circ Z)(t)=Q\bigl(\tilde\rho(t),y(t)\bigr)=2\langle \tilde\rho(t),y(t)\rangle.
\end{equation}
The level set
\[
\mathcal N:=\{Z\in T^*\widetilde{\mathcal M}:Q(Z)=0\}
=\{(\tilde\rho,y):\langle\tilde\rho,y\rangle=0\}
\]
is the \emph{decisional light cone}\footnote{\cite{Levin2019} define a \textit{cognitive light cone}, which is a different concept in another context.}, or simply the \emph{null cone}.  It contains the zero--residual leaf
\(\{y=0\}\), but is strictly larger: the null constraint \(\Lambda=0\) imposes only the orthogonality
\(\langle\tilde\rho,y\rangle=0\) and does not force \(y=0\).

The interpretation of the null constraint is an \emph{orthogonal splitting of bounded rationality deviation}.  For
\(\tilde\rho\neq0\), write \(r:=\|\tilde\rho\|>0\) and \(\rho:=\tilde\rho/r\in\mathbb S^{n-1}\).  Decompose the residual
into its component parallel to the instantaneous amplitude direction and its orthogonal complement,
\[
y_{\parallel}:=\langle\rho,y\rangle\,\rho,
\qquad
y_{\perp}:=y-y_{\parallel},
\qquad
\langle \rho,y_{\perp}\rangle=0.
\]
Then
\[
\Lambda=2\langle \tilde\rho,y\rangle=2r\,\langle\rho,y\rangle,
\]
so \(\Lambda\) measures exactly the \emph{radial} component \(y_{\parallel}\).  In particular,
\[
Z\in\mathcal N
\quad\Longleftrightarrow\quad
\langle\tilde\rho,y\rangle=0
\quad\Longleftrightarrow\quad
y_{\parallel}=0,
\]
i.e.\ on the null cone the instantaneous least--action deviation is purely tangential to the sphere of radius \(r\).

This distinction is behaviourally meaningful because simplex observables depend only on the ray \(\RR_{>0}\tilde\rho\) via the quadratic readout
\[
q_i=\frac{\tilde\rho_i^2}{\|\tilde\rho\|^2}=\rho_i^2,
\]
whereas the lifted model shows how deviation accumulates in phase space.  The quantity \(\Lambda\) isolates the component of
the residual that is aligned with the current amplitude direction and therefore separates ``radial'' deviation
(\(\Lambda\neq0\)) from ``purely redistributive'' deviation at fixed radius (\(\Lambda=0\)).

A central feature of the hyperbolic lift is that the corresponding episode--level index \(\Lambda\) satisfies an exact
balance law with a nonnegative source term.  Differentiating \eqref{eq:Lambda-def} along a solution and substituting
\(\dot{\tilde\rho}=y+\widehat V\tilde\rho\) and \(\dot y=-\widehat V^{\!\top}y\) from \eqref{eq:ham-system} yields
\[
\dot\Lambda(t)
=2\langle \dot{\tilde\rho}(t),y(t)\rangle+2\langle \tilde\rho(t),\dot y(t)\rangle
=2\|y(t)\|^2,
\]
hence the monotonicity identity
\begin{equation}\label{eq:Lambda-dot-neutral}
	\dot\Lambda(t)=2\|y(t)\|^2\ge 0.
\end{equation}
Integrating \eqref{eq:Lambda-dot-neutral} from \(t_0\) to \(t\) gives
\begin{equation}\label{eq:Lambda-integral-neutral}
	\Lambda(t)=\Lambda(t_0)+2\int_{t_0}^t\|y(\tau)\|^2\,d\tau,
	\qquad t\ge t_0.
\end{equation}
It is therefore natural to interpret \(\Lambda\) as a \emph{null--cone index} whose increment records the cumulative
deviation from the zero--residual leaf.  We name the associated nonnegative accumulation functional
\begin{equation}\label{eq:accum-residual}
	\mathcal A(t;t_0):=\int_{t_0}^t\|y(\tau)\|^2\,d\tau,\qquad t\in[t_0,t_1].
\end{equation}
Since \(y=\dot{\tilde\rho}-\widehat V\tilde\rho\) along any episode, one has
\(L(\tilde\rho,\dot{\tilde\rho})=\tfrac12\|y\|^2\), hence the action accumulated over \([t_0,t]\) is
\[
\int_{t_0}^t L\,d\tau=\tfrac12\,\mathcal A(t;t_0)=\tfrac14\bigl(\Lambda(t)-\Lambda(t_0)\bigr).
\]
Thus \(\Lambda\) is simultaneously a geometric phase--space index and a quantitative record of accumulated least--action
deviation.

The monotonicity of \(\Lambda\) imposes a one--way constraint relative to the null cone.  Fix \(t_0\).  If \(\Lambda(t_0)>0\),
then \(\Lambda(t)>0\) for all \(t\ge t_0\).  If \(\Lambda(t_0)=0\) and \(y(t_0)\neq 0\), then
\(\dot\Lambda(t_0)=2\|y(t_0)\|^2>0\), hence \(\Lambda(t)>0\) for all \(t>t_0\).  If \(\Lambda(t_0)<0\), then \(\Lambda\) can
cross \(0\) at most once for \(t\ge t_0\), and any crossing time \(t_*>t_0\) (if it exists) is uniquely characterised by
\[
\Lambda(t_*)=0
\quad\Longleftrightarrow\quad
\mathcal A(t_*;t_0)=\int_{t_0}^{t_*}\|y(\tau)\|^2\,d\tau=-\frac{\Lambda(t_0)}{2}.
\]
Finally, remaining on the null cone for all forward time is equivalent to vanishing residual: by \eqref{eq:Lambda-dot-neutral},
the condition \(\Lambda(t)\equiv \Lambda(t_0)\) for all \(t\ge t_0\) holds if and only if \(\|y(t)\|\equiv 0\), i.e.\
\(y(t)\equiv 0\), for all \(t\ge t_0\).  In particular, any episode that touches \(\mathcal N\) with nonzero residual departs
immediately into the region \(\{Q>0\}\), and cannot return to \(\{Q<0\}\) for forward time.

It is useful to contrast this lifted, episode--level constraint with the intrinsic simplex decompositions discussed earlier.
On \(\mathcal M\) the preference field \(\beta=\mathcal T(\widehat V)\) admits the canonical splitting
\(\beta=\mathcal U+\mathcal R\) with \(\mathcal U=dU\) and \(d\beta=d\mathcal R\) (Proposition~\ref{prop:beta-decomp}), so
co--utility bounded rationality is detected by the curvature \(d\beta\) and can be probed by loop holonomy.  By contrast,
\(\Lambda(t)=2\langle\tilde\rho(t),y(t)\rangle\) depends on the residual variable \(y\) and is therefore not determined by the
epistemic state \(q(t)\) alone.  Co--utility bounded rationality is thus a property of the epistemic preference \emph{field}
on \(\mathcal M\), whereas the null--cone index \(\Lambda\) is a \emph{trajectory} diagnostic of off--shell episodes in the lifted
model: it quantifies accumulated deviation and enforces the one--way cone--crossing constraint that is characteristic of the
hyperbolic least--action lift.

\begin{remark}[The two--outcome case and Lorentzian signature.]
	\label{rem:SO21-n2}
	For \(n=2\) the hyperbolic (split--complex) packaging admits a direct \((2{+}1)\)--dimensional Lorentzian
	interpretation.  After modding out the gauge shift \(\widehat V\mapsto \widehat V+cI\) (which is invisible
	under simplex projection), the effective traceless generator may be regarded as an element of
	\(\mathfrak{sl}(2,\RR)\), and one has the Lie--algebra isomorphism
	\[
	\mathfrak{sl}(2,\RR)\;\cong\;\mathfrak{so}(2,1),
	\]
	where \(\mathfrak{so}(2,1)\) is the Lorentzian Lie algebra in signature \((2,1)\).
	
	At the group level one has the standard double cover
	\[
	\Spin(2,1)\;\longrightarrow\;\mathrm{SO}^+(2,1),
	\qquad
	\ker=\{\pm 1\},
	\]
	and \(\Spin(2,1)\cong \mathrm{SL}(2,\RR)\).  In geometric--algebra terms, \(\Spin(2,1)\) is realised inside the
	even subalgebra \(\Cl^{+}(2,1)\) and acts on \(\RR^{2,1}\) by the sandwich map \(v\mapsto R v R^{-1}\).
	
	Geometrically, this is the hyperbolic analogue of the Bloch--sphere picture for a complex two--level system.
	In the split--complex representation the normalised, projectivised pure--state manifold for \(n=2\) identifies
	with a two--sheeted hyperboloid in \(\RR^{2,1}\), and changes of (para--unitary) frame act by Lorentz
	transformations \(\mathrm{SO}^+(2,1)\) on this ``hyperbolic Bloch'' geometry.  Thus the two--outcome hyperbolic
	setting carries a canonical Lorentzian \(\mathrm{SO}(2,1)\) symmetry, just as the elliptic (complex) two--outcome
	setting carries \(\mathrm{SO}(3)\). We will elaborate on this shortly in \autoref{subsubsec:born-readout}.
\end{remark}

\subsection{Geometry of readout}
\label{subsec:geometry-readout}

The mechanisms analysed so far are intrinsic either to the epistemic simplex dynamics (e.g.\ non--integrability of the induced
preference form) or to the lifted least--action dynamics (e.g.\ the on--shell channel split and the off--shell null--cone index).
We now isolate a distinct source of bounded rationality that is neither a property of the latent generator alone nor a property
of the simplex geometry alone, but of the \emph{geometric reduction} that connects them: observable choice frequencies are
obtained from amplitudes by a nonlinear readout map defined with respect to a fixed orthonormal readout basis of the latent
amplitude space.  Consequently, two latent episodes that are close (or even identical up to a latent gauge) can induce different
observable drifts if the readout basis is changed, and conversely a fixed latent episode can display context dependence solely
because the projection to probabilities is nonlinear.  Since this effect is already present in the \emph{real} (hyperbolic) LAR formulation, with linear amplitude dynamics together with a quadratic normalisation/squaring readout in an orthonormal basis, it can be analysed without invoking elliptic unitarity or any additional measurement postulates.

We therefore proceed in two steps.  First we record the probability readout in a form that is compatible with both the elliptic
and hyperbolic coherent descriptions by introducing a uniform phase--group notation.  We then study how varying the readout map
across contexts induces distinct observable drifts and produces the familiar cross--term mechanism usually called
(interference) when multiple amplitude contributions are coherently combined prior to quadratic normalisation.

\subsubsection{Probability readout, phase groups, and ambient signature}
\label{subsubsec:born-readout}

The epistemic readout is defined directly from the real amplitude variable \(\tilde\rho\). For \(\tilde\rho\neq 0\), set
\begin{equation}\label{eq:born-readout}
	\mathcal Z(\tilde\rho):=\|\tilde\rho\|^2=\sum_{k=1}^n(\tilde\rho_k)^2,
	\qquad
	q_i(\tilde\rho,y):=\frac{(\tilde\rho_i)^2}{\mathcal Z(\tilde\rho)}.
\end{equation}
Thus \(q(\tilde\rho,y)\) depends on \((\tilde\rho,y)\in N\) only through \(\tilde\rho\in\RR^n\); the residual variable \(y\)
affects \(q\) only through the induced trajectory \(t\mapsto(\tilde\rho(t),y(t))\).  The map \eqref{eq:born-readout} is
homogeneous of degree \(0\), hence invariant under \(\tilde\rho\mapsto\lambda\tilde\rho\) for \(\lambda\neq 0\), and it is
sign--blind in each component: \(\tilde\rho_i\mapsto-\tilde\rho_i\) leaves \(q\) unchanged.  Boundary events \(q_i=0\)
are equivalent to \(\tilde\rho_i=0\).

For notational comparison between the elliptic and hyperbolic sectors, let
\[
U(\CC):=\{u\in\CC^\times:\overline u\,u=1\}\cong S^1,
\qquad
U(\DD):=\{u\in\DD^\times:u^\#u=1\}.
\]
Here \(U(\DD)\) is the unit hyperbola, with two connected components; any choice of component is irrelevant for the readout,
since only the relation \(u^\#u=1\) enters below.  Given \(\tilde\rho\in\RR^n\setminus\{0\}\) and
\(u=(u_1,\dots,u_n)\in U(\mathbb K)^n\), define the associated \(\mathbb K\)--valued vector
\begin{equation}\label{eq:psi-from-rho}
	\psi=\psi(\tilde\rho;u)\in\mathbb K^n,
	\qquad
	\psi_i:=\tilde\rho_i\,u_i,
	\qquad i=1,\dots,n.
\end{equation}
Then, componentwise,
\[
\psi_i^*\psi_i=(\tilde\rho_i)^2
\quad (\mathbb K=\CC),
\qquad
\psi_i^\#\psi_i=(\tilde\rho_i)^2
\quad (\mathbb K=\DD),
\]
and therefore
\[
\sum_{k=1}^n\psi_k^*\psi_k=\mathcal Z(\tilde\rho)
\quad (\mathbb K=\CC),
\qquad
\sum_{k=1}^n\psi_k^\#\psi_k=\mathcal Z(\tilde\rho)
\quad (\mathbb K=\DD).
\]

Hence the quadratic readout \eqref{eq:born-readout} admits the the following Born-form identity
\begin{equation}\label{eq:born-form-K}
	q_i
	=
	\frac{\psi_i^*\psi_i}{\sum_{k=1}^n\psi_k^*\psi_k}
	\quad (\mathbb K=\CC),
	\qquad
	q_i
	=
	\frac{\psi_i^\#\psi_i}{\sum_{k=1}^n\psi_k^\#\psi_k}
	\quad (\mathbb K=\DD).
\end{equation}
Equation \eqref{eq:born-form-K} is an identity, not an additional postulate. In the elliptic sector, \cite{Caticha2021} likewise derives the Born-form expression rather than postulating it. Here the probabilities are defined by \eqref{eq:born-readout}, and \eqref{eq:born-form-K} merely expresses this quadratic readout as normalised component intensities for the vectors \(\psi\) of the form \eqref{eq:psi-from-rho}. In particular, \eqref{eq:born-form-K} is \emph{not} asserted for arbitrary \(\psi\in\mathbb K^n\).

On the chosen holomorphic eigenbundle one may take \(\psi=\Psi\in\CC^n\), and on the chosen para--holomorphic eigenbundle one
may take \(\psi=\Psi_{\DD}\in\DD^n\).  In both cases the probability rule is still \eqref{eq:born-readout}; only its
notation changes.  The elliptic/hyperbolic distinction therefore lies not in the quadratic readout
\eqref{eq:born-readout}--\eqref{eq:born-form-K}, but in the phase geometry and in the corresponding \(\mathbb K\)--Hermitian
packaging of the state evolution (cf.\ \eqref{eq:schrödinger}, \eqref{eq:para-schrodinger}).  For \(n=2\), this is the
Lorentzian \(\mathrm{SO}(2,1)\) carrier of Remark~\autoref{rem:SO21-n2}.

\begin{remark}[Hyperbolic Born form in the readout basis]
	\label{rem:hyperbolic-born-readout-basis}
	Hyperbolic quantum mechanics over \(\DD\) (e.g.\ \cite{Khrennikov2000hyperbolicquantummechanics}) often writes probabilities as
	\[
	q_i=\frac{c_i^\#c_i}{\sum_{k=1}^n c_k^\#c_k},
	\]
	with \(c_i\in\DD\) the expansion coefficients of a split--complex state in a chosen readout basis.  For a general
	\(\DD\)--vector, the quantities \(c_i^\#c_i\) may be indefinite.  In the present framework, however, the coefficients entering
	the readout are constrained to be of the form
	\[
	c_i=(\Psi_{\DD})_i=\tilde\rho_i u_i,
	\qquad
	u_i\in U(\DD),
	\qquad
	u_i^\#u_i=1.
	\]
	It follows immediately that
	\[
	c_i^\#c_i=(\tilde\rho_i)^2\ge 0,
	\]
	with no sign restriction on \(\tilde\rho_i\).  Thus the hyperbolic formula in \eqref{eq:born-form-K} is exactly equivalent to
	the defining readout \eqref{eq:born-readout}.
\end{remark}

\subsubsection{Contexts and hyperbolic interference}
\label{subsubsec:Contexts-hyperbolic-interference}

Having fixed the quadratic readout and its phase--group notation above, we now allow the readout itself
to vary.  A \emph{context} is specified by a choice of readout basis (equivalently, by a choice of
readout map), and the same latent amplitude state can therefore produce different observable probability
assignments in different contexts.  Because the readout is nonlinear, its differential depends on the
chosen context, so even when the latent dynamics is context--independent the induced simplex drift need
not be.  This readout--geometric dependence is the basic contextuality mechanism in our framework.

Recall that latent states are represented by normalised amplitudes
\(\rho\in\mathbb S^{n-1}\subset\RR^n\), and that observable probabilities arise from the quadratic
readout.  The canonical readout corresponds to the standard basis, i.e.\ \(q_k=\rho_k^2\).
To model contextual readout, we allow the readout basis to vary: a \emph{context} is an orthogonal matrix
\(B\in O(n)\) with columns \(b_1,\dots,b_n\).  The associated readout map is
\begin{equation}\label{eq:readout-piB}
	\pi_B:\mathbb S^{n-1}\to\Delta^{n-1},
	\qquad
	\pi_B(\rho)_k:=\langle b_k,\rho\rangle^2,
\end{equation}
so that \(\sum_k \pi_B(\rho)_k=\|\rho\|^2=1\).  The canonical context is \(B=I\), giving
\(\pi_I(\rho)_k=\rho_k^2\).  Equivalently, writing \(\rho^{(B)}:=B^\top\rho\), one has
\(\pi_B(\rho)_k=(\rho^{(B)}_k)^2\), i.e.\ the same quadratic readout applied in rotated coordinates; in
particular, for real contexts \(B\in O(n)\) the split--complex form \eqref{eq:born-form-K} remains
available by setting \(c^{(B)}_k:=\rho^{(B)}_k u_k\) with \(u_k\in U(\DD)\), so that
\[
(c^{(B)}_k)^*c^{(B)}_k=(\rho^{(B)}_k)^2=\pi_B(\rho)_k.
\]

Now let \(X\) denote the latent LAR vector field \eqref{eq:LAR-vector-field} on \(\mathbb S^{n-1}\).
In context \(B\), write \(q:=\pi_B(\rho)\in\Delta^{n-1}\).  The induced \emph{instantaneous observable
	drift} at \(\rho\) is the pushforward \(d\pi_{B,\rho}(X(\rho))\), which satisfies \(\sum_k \dot q_k=0\)
and hence defines an admissible simplex drift whenever \(q\in\mathcal M\).  Since \(\pi_B\) is
nonlinear, its differential depends on \(B\), so the observable simplex trajectory depends on which
readout basis \(B\) is used to read out the same latent amplitude episode.

This does not exhaust the effect of changing basis.  Even for a fixed context \(B\), the same quadratic
readout can produce cross terms whenever the \(B\)--coordinates of the latent state contain several
modal contributions before readout.  In that case squaring generates off--diagonal terms.  In
particular, this occurs whenever the chosen readout basis is not aligned with the spectral basis of the
on--shell generator.  In the split--complex extension \(\DD\), the same algebraic mechanism yields the
hyperbolic interference laws studied in the hyperbolic quantum literature
(e.g.\ \citep{Khrennikov2000hyperbolicquantummechanics}).

To make this explicit, on the zero--residual leaf one may use the auxiliary complex packaging
\eqref{eq:schrodinger-complex-packaging}.  The corresponding effective complex Hamiltonian
\(\widehat H_{\mathbb C}=i\,\widehat V\) is generically non--Hermitian whenever \(\widehat S\neq 0\).
Assuming \(\widehat V\) is diagonalizable over \(\mathbb C\), write
\[
\widehat V=P\Lambda P^{-1}.
\]
Then
\[
z_i(t)=\sum_a P_{ia}e^{\lambda_a t}c_a,
\qquad
w_i(t):=z_i(t)^2
=
\sum_{a,b}P_{ia}P_{ib}\,c_a c_b\,e^{(\lambda_a+\lambda_b)t}.
\]
The terms with \(a\neq b\) are precisely the interference terms.  Since \(\widehat V\) and the
on--shell initial data are real, the complex modal contributions recombine to a real trajectory
\(z(t)=\tilde\rho(t)\in\RR^n\).  The observable probabilities are therefore still computed by the
same readout, written equivalently in the form \eqref{eq:born-form-K} for representatives of the
type \eqref{eq:psi-from-rho}.

Thus, a misaligned context can produce interference, but the two notions should not be
identified: context refers to the chosen readout basis, whereas interference refers to the
off--diagonal cross terms generated within that fixed basis by the quadratic readout. In the real hyperbolic LAR
framework, both are intrinsic consequences of linear amplitude dynamics combined with the quadratic
readout map to the simplex, and neither requires complex Hilbert space, elliptic unitarity, or any
additional measurement postulate.

\subsection{Bounded rationality in elliptic quantum geometry}
\label{subsec:BR-elliptic-quantum}

So far in \autoref{sec:BR} we have worked in the hyperbolic information--geometric regime: a real para--K\"ahler LAR lift on
the doubled phase space \(N=T^*\widetilde{\mathcal M}\), and bounded rationality arising when the latent flow is reduced to the
epistemic simplex.  In that setting the latent split \(\widehat V=\widehat S+\widehat F\) separates evaluative and co--utility
channels, and a specifically hyperbolic, lifted diagnostic is available off--shell: the neutral form \(Q\) (and its index
\(\Lambda\)) controls residual accumulation through the doubled variable \(y\).

By contrast, the elliptic (complex, Hilbert--space) quantum model arises in our framework by imposing the coherent admissibility
restriction of Coherent Least--Action Rationality (CLAR).  One passes to the complexification \(N_\CC\) and restricts to a
chosen complex Lagrangian splitting \(TN_\CC=L^{(+)}\oplus L^{(-)}\) (cf.\ \autoref{subsec:CLAR} and \autoref{sec:4}).  On the
resulting holomorphic sector the reduced dynamics closes as the standard Schr\"odinger evolution with Hermitian packaged
generator  in (\ref{eq:schrödinger})
\[
\widehat H\;:=\;\widehat S+i\widehat F,
\]
hence is unitary with respect to the induced Hilbert inner product.  In particular, the Hilbert norm \(\|\psi(t)\|\) is
conserved along CLAR trajectories.  From the viewpoint of the present section, this means that the lifted, off--shell hyperbolic
mechanism of \autoref{subsec:neutral-cone} does not play the same role in the elliptic regime: the neutral index \(\Lambda\) is a
trajectory diagnostic tied to the real para--K\"ahler lift on the doubled variables \((\tilde\rho,y)\), whereas CLAR replaces the
full lifted flow by a polarization--projected, coherence--preserving evolution on \(N_\CC\).  
Nevertheless, bounded rationality in the elliptic regime still persists through epistemic non--integrability (curvature of the induced preference form) and through readout--induced contextuality, since both arise from the reduction from latent amplitudes to observable probabilities rather than from the off--shell hyperbolic lift.

To make this persistence precise, we now pass from the latent unitary CLAR evolution to its induced epistemic dynamics in a fixed readout context.  Since the off--shell lifted diagnostic \(\Lambda\) is absent in the elliptic regime, the relevant bounded--rationality mechanisms are those that survive the reduction to observable probabilities: the induced simplex drift, its Fisher--Rao dual preference form, and their dependence on the chosen readout context.

As in \autoref{subsec:geometry-readout}, fix an admissible readout context \(B=\{b_1,\dots,b_n\}\subset\RR^n\), i.e.\ the orthonormal basis used in the quadratic readout.  
In the elliptic model the
latent state is a complex amplitude \(\psi(t)\in\CC^n\setminus\{0\}\) evolving unitarily on a coherent sector, and the associated
epistemic state is obtained by the modulus--squared readout in the chosen basis,
\[
q_B(t)=\pi_B(\psi(t))\in\mathcal M,
\qquad
\pi_B(\psi)_k:=\frac{|\langle b_k,\psi\rangle|^2}{\|\psi\|^2},
\]
where we restrict attention to episodes for which \(q_B(t)\in\mathcal M=\Delta^{n-1}_{>0}\) (so Fisher--Rao geometry applies
without boundary issues).  For unitary CLAR trajectories \(\|\psi(t)\|\) is constant, so the normalisation is inessential, but we
keep it explicit to emphasise ray--invariance.

The induced simplex drift in context \(B\) is obtained by pushing forward the latent velocity,
\[
\dot q_B(t)=d\pi_{B,\psi(t)}\bigl(\dot\psi(t)\bigr),
\]
which is automatically tangent to the simplex hyperplane \(\sum_k q_k=1\).  Since \(\pi_B\) is nonlinear and not injective,
\(\dot q_B(t)\) generally depends on latent phase information in \(\psi(t)\).  Accordingly, the reduced drift is naturally viewed
as an object defined \emph{along a trajectory} rather than as an autonomous vector field on \(\mathcal M\).

To relate back to the epistemic analysis of \autoref{subsec:non-integrability}, one may associate to the reduced drift a
(preference) covector along the trajectory by Fisher--Rao duality,
\[
\beta_B(t)\;:=\;g_F\!\bigl(\dot q_B(t),\,\cdot\,\bigr)\in T^*_{q_B(t)}\mathcal M.
\]
Whenever the reduced drift admits a smooth representative one--form field \(\beta_B\in\Omega^1(\mathcal M)\) in a fixed context
(e.g.\ after choosing a coherent gauge/section on the holomorphic sector), the same integrability question arises as in the
hyperbolic case: the curvature \(d\beta_B\) detects epistemic path dependence in that context, while variation of \(B\) probes the
readout--geometric contextuality mechanism of \autoref{subsec:geometry-readout}.

\begin{definition}[Coherent bounded rationality in a fixed context]
	Assume the latent dynamics satisfy CLAR (so \(\psi(t)\) evolves unitarily on a holomorphic sector), and fix a readout context
	\(B\) for which the reduced drift admits a representative preference form \(\beta_B\in\Omega^1(\mathcal M)\).
	Assume moreover that \(\beta_B\) admits the canonical decomposition
	\[
	\beta_B=dU_B+\mathcal R_B
	\]
	as in Proposition~\ref{prop:beta-decomp}.  We say that the epistemic dynamics in context \(B\) exhibit
	\emph{coherent bounded rationality} if along every CLAR trajectory the evaluative rate vanishes,
	\[
	\frac{d}{dt}U_B\bigl(q_B(t)\bigr)=dU_B\bigl(\dot q_B(t)\bigr)=0.
	\]
	Equivalently, \(g_F\!\bigl(\grad_{g_F}U_B(q_B(t)),\,\dot q_B(t)\bigr)=0\) whenever \(q_B(t)\in\mathcal M\).
\end{definition}

Thus, in the elliptic regime, CLAR provides coherence--preserving (unitary) latent motion generated by the Hermitian operator
\(\widehat H=\widehat S+i\widehat F\), while epistemic bounded rationality is determined by how this motion is reduced to
observable probabilities by the nonlinear readout family \(\pi_B\). 
In fact, this is precisely the setting of the complex Hilbert--space quantum--cognition literature: bounded rationality is
compatible with coherence--preserving (unitary) latent evolution, and arises through geometric reduction, namely the nonlinear,
context--dependent readout family \(\pi_B\) that maps latent states to observable probabilities on the simplex.
What our framework adds in the hyperbolic (para--K\"ahler) regime is an additional, genuinely lifted mechanism that is invisible
under coherent elliptic restriction: the full least--action dynamics on the doubled variables \((\tilde\rho,y)\) admits the
neutral index \(\Lambda(t)=2\langle\tilde\rho(t),y(t)\rangle\), which provides an intrinsic accumulation diagnostic and a one--way
cone--crossing constraint for residual deviation.  This \(\Lambda\)--mechanism has no analogue in the unitary CLAR sector, where
one restricts attention to coherent holomorphic evolution and does not track episodes outside the non-zero residual leaf.

\section{Conclusions and Further Research}
\label{sec:conclusions}

The role of information theory in behavioural choice modelling has been mixed.
\citet{Luce2003} argued that, despite its early appeal, information theory offered
limited explanatory value for psychological theories of choice.  His critique was
structural: standard information--theoretic formulations treat alternatives as
unstructured elements of a finite set, whereas behavioural stimuli typically
come with internal relations (similarity, salience, or metric structure) that
shape observed choice patterns.  In this view, entropy--based approaches risk
capturing only a generic ``softening'' of choice, without explaining how context
and internal structure enter.

Interestingly, the same year,  Shannon entropy was reintroduced into models
of bounded rationality in a different role:  \citet{Sims2003} treated entropy not
as a description of stimuli but as a \emph{resource cost} governing internal
representation.  This reinterpretation sidesteps Luce's objection about
unstructured alternatives, yet it still operates directly on the simplex and
does not by itself supply an intrinsic geometry linking the latent layer of
representation to the observable layer of choice.

\paragraph{Main contribution.}
We develop a unified information--geometric theory of choice in which observable probabilities on the simplex are obtained by a
quadratic (Born--type) readout from a latent amplitude dynamics governed by a single \emph{preference operator}
\(\widehat V\).  The operator admits the canonical decomposition \(\widehat V=\widehat S+\widehat F\), and this split has direct
behavioural content.  The symmetric sector \(\widehat S\) is the \emph{evaluative} channel: it drives welfare/utility production
through norm change on rays, induces gradient--type drift, and generates a canonical monotone \emph{entropic welfare clock} given
by the log--partition normaliser.  This is the same log--sum--exp (inclusive--value) potential that underlies
entropy--regularised choice, free--energy formulations, and rational--inattention objective functions, and it is canonically
paired with Shannon entropy by convex duality.  The skew sector \(\widehat F\) is the \emph{circulatory} channel: it generates
redistributive drift without welfare production and induces co--utility/regret via non--integrable (holonomic) components of the
epistemic preference field.

Crucially, \(\widehat F\) identifies a genuine \emph{split--complex} (para--complex) sector of the lifted dynamics.  On the
cotangent lift \(N=T^*\widetilde{\mathcal M}\), equipped with its canonical para--K\"ahler structure, we derive the Hamiltonian
least--action equations and show that on the distinguished zero--residual leaf \(y=0\) the dynamics closes as a linear
\emph{para--Schr\"odinger} evolution.  This provides a principled motivation for working with real amplitudes in quantum--like
cognition models: the canonical hyperbolic lift already contains a coherent Schr\"odinger--type sector with a transparent
welfare--vs--circulation interpretation, without assuming complex Hilbert space.

The framework extends beyond the coherent leaf.  Away from \(y=0\), bounded rationality is organised by intrinsic
phase--space diagnostics of the hyperbolic lift, including the neutral pairing and its null--cone index, which record
accumulated least--action deviation and enforce one--way cone--crossing constraints for non-zero residual episodes.  Taken together, the
theory delivers an operator--geometric foundation in which welfare production, regret/co--utility, coherent split--complex
dynamics, arise from a single least--action principle, even outside the non-zero residual leaf.

These constructions also clarify the relationship between entropy--regularised choice and quantum--like cognition models.  The
approach is explicitly \emph{probabilities--first}, in the spirit of information--geometric reconstructions, but it does not
begin by postulating a complex Hilbert space.  Instead, the Fisher--Rao/Hessian data determine a canonical para--K\"ahler
structure on the lifted manifold, and observable probabilities are obtained directly on the real leaf via the same quadratic
normalisation.

At the same time, the familiar elliptic (Hilbert--space) regime can be accommodated \emph{within the same information--geometric
origin}.  What changes is not the underlying symplectic structure \((N,\Omega)\) but the class of admissible variations: an
elliptic model is obtained only after supplying additional admissibility input that selects an \(\Omega\)--compatible complex
Lagrangian sector \(L^{(+)}\subset TN_\CC\) and restricts trajectories accordingly.  Within such a sector the restricted dynamics
closes as a unitary Schr\"odinger--type evolution, and Born--type probabilities again arise from the same quadratic readout.  In
this precise sense, elliptic coherence is not a primitive consequence of least--action rationality; it is a coherent restriction
of the canonical hyperbolic lift, motivating coherent least--action rationality (CLAR) as a least--action postulate formulated
\emph{within} the selected complex sector.

Several directions for further research follow naturally.  On the applied side,
the operator--based (pairwise) structure suggests new estimators and model
classes that interpolate between additive utility, pairwise contextual
couplings, and regret--driven circulation, with clear links to existing
entropy--regularised and context--dependent choice models.  On the geometric
side, it would be valuable to characterise when (and in what sense) admissible
complex sectors can be selected in a data--driven way, and to understand how the
hyperbolic null--cone and intrinsic clocks interact with projection, aggregation
and measurement schemes on the simplex.  A related direction is to study
hybrid and open--system perturbations that interpolate between coherence--restricted
and unrestricted lifted dynamics (e.g.\ channel/Lindblad--type effects in
quantum--like decision dynamics; \citep[e.g.][]{AsanoEtAl2012}).

\setlength{\bibsep}{0pt plus 0.3ex}

\end{document}